\documentclass[aps,pra,showpacs,twocolumn,twoside,superscriptaddress]{revtex4-1}

\usepackage[utf8]{inputenc}
\usepackage[T1]{fontenc}
\usepackage[english]{babel}  

\usepackage{times}

\usepackage{amssymb,amsmath,amsfonts,amsthm}
\usepackage{mathtools}
\usepackage{verbatim}
\usepackage{enumerate}
\usepackage{graphicx}
\graphicspath{{pics/}}
\usepackage{hyperref}
\usepackage{bbm}
\usepackage{booktabs}
\usepackage{cleveref}

\usepackage[caption=false]{subfig}

\usepackage{color}
\definecolor{dred}{rgb}{.8,0.2,.2}
\definecolor{ddred}{rgb}{.8,0.5,.5}
\definecolor{dblue}{rgb}{.2,0.2,.8}
\definecolor{dgreen}{rgb}{.2,0.5,.2}

\newcommand{\ket}[1]{\ensuremath{|#1\rangle}}
\newcommand{\bra}[1]{\ensuremath{\langle#1|}}

\newcommand{\tr}{\mathop{\rm tr}\nolimits}

\def\B{\mathcal{B}}

\newtheorem{definition}{Definition}
\newtheorem{lemma}{Lemma}
\newtheorem{theorem}{Theorem}
\newtheorem{corollary}{Corollary}

\newcommand{\nc}{\newcommand}
\nc{\cA}{{\cal A}} \nc{\cB}{{\cal B}} \nc{\cC}{{\cal C}}
\nc{\cD}{{\cal D}} \nc{\cE}{{\cal E}} \nc{\cF}{{\cal F}}
\nc{\cG}{{\cal G}} \nc{\cH}{{\cal H}} \nc{\cI}{{\cal I}}
\nc{\cJ}{{\cal J}} \nc{\cK}{{\cal K}} \nc{\cL}{{\cal L}}
\nc{\cM}{{\cal M}} \nc{\cN}{{\cal N}} \nc{\cO}{{\cal O}}
\nc{\cP}{{\cal P}} \nc{\cQ}{{\cal Q}} \nc{\cR}{{\cal R}}
\nc{\cS}{{\cal S}} \nc{\cT}{{\cal T}} \nc{\cU}{{\cal U}}
\nc{\cV}{{\cal V}} \nc{\cW}{{\cal W}} \nc{\cX}{{\cal X}}
\nc{\cZ}{{\cal Z}}

\begin{document}


\title{Pure State Tomography with Pauli Measurements}

\author{Xian Ma}%
\affiliation{Institute for Quantum Computing,
University of Waterloo, Waterloo, Ontario, N2L 3G1, Canada}%
\affiliation{Department of Physics and Astronomy, University of Waterloo, Waterloo, Ontario N2L 3G1, Canada}
\author{Tyler Jackson}%
\affiliation{Department of Mathematics \& Statistics, University of
  Guelph, Guelph, Ontario, N1G 2W1, Canada}
\affiliation{Institute for Quantum Computing,
University of Waterloo, Waterloo, Ontario,  N2L 3G1, Canada}%
\author{Hui Zhou}%
\affiliation{Hefei National Laboratory for Physical Sciences at Microscale and Department of Modern Physics, University of Science
and Technology of China, Hefei, Anhui 230036, China}%
\affiliation{Synergetic Innovation Center of Quantum Information \& Quantum Physics,
University of Science and Technology of China, Hefei, Anhui 230026, China}%
\author{Jianxin Chen} %
\address{Joint Center for Quantum Information and Computer Science,
  University of Maryland, College Park, Maryland, USA}
\author{Dawei Lu}%
\affiliation{Institute for Quantum Computing,
University of Waterloo, Waterloo, Ontario, N2L 3G1, Canada}%
\affiliation{Department of Physics and Astronomy, University of Waterloo, Waterloo, Ontario N2L 3G1, Canada}
\author{Michael D. Mazurek}%
\affiliation{Institute for Quantum Computing,
University of Waterloo, Waterloo, Ontario,  N2L 3G1, Canada}%
\affiliation{Department of Physics and Astronomy, University of Waterloo, Waterloo, Ontario N2L 3G1, Canada}
\author{Kent A. G. Fisher}%
\affiliation{Institute for Quantum Computing,
University of Waterloo, Waterloo, Ontario,  N2L 3G1, Canada}%
\affiliation{Department of Physics and Astronomy, University of Waterloo, Waterloo, Ontario N2L 3G1, Canada}
\author{Xinhua Peng}%
\affiliation{Hefei National Laboratory for Physical Sciences at Microscale and Department of Modern Physics, University of Science
and Technology of China, Hefei, Anhui 230036, China}%
\affiliation{Synergetic Innovation Center of Quantum Information \& Quantum Physics,
University of Science and Technology of China, Hefei, Anhui 230026, China}%
\affiliation{Institute for Quantum Computing, University of Waterloo,
  Waterloo, Ontario, N2L 3G1, Canada}%
\author{David Kribs}%
\affiliation{Department of Mathematics \& Statistics, University of
  Guelph, Guelph, Ontario, N1G 2W1, Canada}
\affiliation{Institute for Quantum Computing,
University of Waterloo, Waterloo, Ontario, N2L 3G1, Canada}%
\author{Kevin J. Resch}%
\affiliation{Institute for Quantum Computing,
University of Waterloo, Waterloo, Ontario, N2L 3G1, Canada}%
\affiliation{Department of Physics and Astronomy, University of Waterloo, Waterloo, Ontario N2L 3G1, Canada}
\author{Zhengfeng Ji}%
\affiliation{Institute for Quantum Computing,
University of Waterloo, Waterloo, Ontario, N2L 3G1, Canada}%
\author{Bei Zeng}%
\affiliation{Department of Mathematics \& Statistics, University of
  Guelph, Guelph, Ontario, N1G 2W1, Canada}
\affiliation{Institute for Quantum Computing,
University of Waterloo, Waterloo, Ontario, N2L 3G1, Canada}%
\affiliation{Canadian Institute for Advanced Research, Toronto,
  Ontario, M5G 1Z8, Canada}
\author{Raymond Laflamme}%
\affiliation{Institute for Quantum Computing,
University of Waterloo, Waterloo, Ontario,  N2L 3G1, Canada}
\affiliation{Department of Physics and Astronomy, University of Waterloo, Waterloo, Ontario N2L 3G1, Canada}
\affiliation{Canadian Institute for Advanced Research, Toronto,
  Ontario, M5G 1Z8, Canada}
\affiliation{Perimeter Institute for Theoretical Physics, Waterloo, Ontario, N2L 2Y5,
Canada}%

\begin{abstract}
We examine the problem of finding the minimum number of
Pauli measurements needed to uniquely determine
an arbitrary $n$-qubit pure state among all quantum states.
We show that only $11$ Pauli
measurements are needed to determine an arbitrary
two-qubit pure state compared
to the full quantum state tomography with $16$ measurements,
and only
$31$ Pauli measurements are needed to determine an arbitrary
three-qubit pure state compared
to the full quantum state tomography with $64$ measurements. We demonstrate that our protocol is robust under depolarizing error with simulated random pure states. We experimentally test the protocol on two- and three-qubit systems with nuclear magnetic resonance techniques. We show that the pure state tomography
protocol saves us a number of measurements without considerable
loss of fidelity. We compare our protocol with same-size sets of
randomly selected Pauli operators and find that our selected set of Pauli
measurements significantly outperforms
those random sampling sets. As a direct application, our scheme can also be used to reduce the number of settings needed for pure-state tomography in quantum optical systems.
\end{abstract}

\date{\today}


\maketitle

\section{Introduction}

We consider a $d$-dimensional Hilbert space $\mathcal{H}_d$,
and denote $D(\mathcal{H}_d)$ the set of density operators acting on $\mathcal{H}_d$.
Assume that we measure a set of $m$ linearly independent observables
\begin{equation}
\label{eq:A}
\mathbf{A}=(A_0,A_1,A_2,\ldots,A_{m-1}),
\end{equation}
where each $A_i$ is Hermitian. Without loss of generality, we assume
$A_0={I}$ (i.e. the identity operator on $\mathcal{H}_d$), and $\tr A_i=0$ for $i=1,2,\ldots,m-1$.

Then for any $\rho\in D(\mathcal{H}_d)$,
the measurement returns a set of outcomes
\begin{equation}
\label{eq:A}
\boldsymbol{\alpha}=(\tr{\rho},\tr(\rho A_1),\tr(\rho A_2),\ldots,\tr(\rho A_{m-1})).
\end{equation}
Theoretically, we always have $\tr\rho=1$, however we keep this entry in
$\boldsymbol{\alpha}$ for the reason of experimental calibration~\cite{james01,deBurgh05,vandersypen2005nmr,oliveira2011nmr}.

For any $\rho\in D(\mathcal{H}_d)$, full quantum state tomography requires
$d^2$ measurement outcomes to determine $\rho$~\cite{NC00}. However
for a pure state $\ket{\psi}\in\mathcal{H}_d$, in general only order $d$
measurements are needed to determine $\ket{\psi}$. There is a slight difference
in interpreting the term `determine', as clarified in~\cite{chen2013uniqueness} and
summarized in the following definition. The
physical interpretation in this case is clear: it is useful in quantum tomography to have some prior knowledge that the state to be reconstructed is pure or nearly pure.
\begin{definition}
A pure state $\ket{\psi}$ is \emph{uniquely determined among pure states (UDP) by measuring $\mathbf{A}$} if there
does not exist any other pure state which has the same measurement results as those of $\ket{\psi}$ when measuring $\mathbf{A}$. A pure state $\ket{\psi}$ is \emph{uniquely determined among all states (UDA) by measuring $\mathbf{A}$} if there
does not exist any other state, pure or mixed, which has the same measurement results as those of $\ket{\psi}$ when measuring $\mathbf{A}$.
\end{definition}

It is known that there exists a family of $4d-4$ observables such that any $d$-dimensional pure state is UDP~\cite{HTW11}, and $5d-6$ observables such that any $d$-dimensional pure state is UDA~\cite{chen2013uniqueness}. Many other techniques for pure-state tomography have been developed, and experiments have been performed to demonstrate the reduction of the number of measurements needed~\cite{Wei92,AW99,Fin04,FSC05,li2015fisher,baldwin2015informational}. However, even if there are constructive protocols for the measurement set $\mathbf{A}$, in practice these sets may not be easy to measure in an experiment.

One idea of the compressed sensing protocols as discussed in~\cite{GLF+10,CPF+10} considers measurements of Pauli operators for
$n$-qubit systems, with Hilbert space dimension $d=2^n$. Since no joint measurements on multiple qubits are needed for Pauli operators, these operators are relatively easy to measure in practice.
It is shown that order $d\log d$ random Pauli measurements are sufficient to UDP almost
all pure states. That is, all pure states can be determined, up to a set of states with measure zero (i.e. `almost all' pure states are determined). Experiments also demonstrate the usefulness of this method in pure-state tomography in practice~\cite{LZL+12}. However, it remains open how many Pauli measurements are needed to determine all pure states (UDP or UDA) of an $n$-qubit system.

In this work, we examine the problem of the minimum number of Pauli operators needed to UDA all $n$-qubit pure states. For $n=1$ the number is known to be $3$, i.e. all three Pauli operators $X,Y,Z$ are needed. We solve the problem for $n=2$ and $n=3$, where at least $11$ Pauli operators are needed for $n=2$ and at least $31$ Pauli operators are needed for $n=3$. We then demonstrate that our protocol is robust under depolarizing error with simulated random pure states. We further implement our
protocol in our nuclear magnetic resonance (NMR) system and compare our result with other methods. As a direct application of this result, we show that our scheme can also be used to reduce the number of settings needed for pure-state tomography in quantum optical systems.

\section{Pure-state tomography using Pauli operators }

We consider the real span of the operators in $\mathbf{A}$, and denote it by $\mathcal{S}(\mathbf{A})$.
Let $\mathcal{S}(\mathbf{A})^{\perp}$ be the $(d^2-m)$-dimensional orthogonal complement subspace of $\mathcal{S}(\mathbf{A})$ inside $\mathbb{R}^{d^2}$.
It is known that a sufficient condition for any pure state $|\phi\rangle$ to be UDA by measuring $\mathbf{A}$ is that any nonzero Hermitian operator $H\in(\mathcal{S}(\mathbf{A}))^{\perp}$
have at least two positive and two negative eigenvalues~\cite{chen2013uniqueness}.
In fact, this is also a necessary condition. Otherwise, if the second-lowest eigenvalue of $H$ is non-negative, then the two states $\ket{\psi}\bra{\psi}$ and $H+\ket{\psi}\bra{\psi}$ are indistinguishable by only measuring $\mathbf{A}$ where $\ket{\psi}$ is the eigenvector of $H$ corresponds to the smallest eigenvalue. Note that without loss of generality, we can always assume the smallest eigenvalue of $H$ is greater than $-1$ which will guarantee $H+\ket{\psi}\bra{\psi}\geq 0$. A similar argument holds if the second-largest eigenvalue of $H$ is non-negative. We will then look for such sets $\mathbf{A}$ containing only Pauli operators, for two-qubit and three-qubit pure state tomography.

\subsection{Two-qubit system}

We denote the single-qubit Pauli operators by $\sigma_1=X,\sigma_2=Y,\sigma_3=Z$, and the identity
operator $\sigma_0=I$.  For a single qubit, it is straightforward to check that measuring
only two of the three operators cannot determine an arbitrary pure
state. Therefore all three Pauli operators are needed in the single-qubit case.

For the two-qubit system, there are a total of $16$ Pauli operators, including the identity. These are given by
the set $\{\sigma_i\otimes\sigma_j\}$ with $i,j=0,1,2,3$. For simplicity we omit the tensor product symbol by writing, e.g. $XY$ instead of $X\otimes Y$.
Of these $16$ Pauli operators, there exists a set of $11$ Pauli operators $\mathbf{A}$ such that $\mathbf{A}$ is UDA for any pure state, as given by the following theorem.

\begin{theorem}
\label{th:2qubits}
Any two-qubit pure state $|\phi\rangle$ is UDA by measuring the following set of Pauli operators.

\begin{eqnarray}
\label{eq:2qubit}
\mathbf{A}=&&\{II,IX,IY,IZ,XI,YX,\nonumber\\
&&YY,YZ,ZX,ZY,ZZ\},
\end{eqnarray}
and no set with fewer than $11$ Pauli operators can be UDA for
all two-qubit pure states.
\end{theorem}

This is to say, $11$ is the minimum number of Pauli operators
needed to UDA any two-qubit pure state, and an example of such a set
with
$11$ Pauli operators is given in Eq.~\eqref{eq:2qubit}.

\begin{proof}
In order for $\mathbf{A}$ to UDA all two-qubit pure states it is known~\cite{chen2013uniqueness} that any Hermitian operator $H\in(\mathcal{S}(\mathbf{A}))^{\perp}$ must
have at least two positive and two negative eigenvalues.

In this case $(\mathcal{S}(\mathbf{A}))^{\perp}=\mathcal{S}(\{XX,XY,XZ,YI,ZI\})$. Note that the $5$ operators which are not measured all mutually anti-commute with each other. It is easy to see that this property is required for if two operators in $(\mathcal{S}(\mathbf{A}))^{\perp}$ commuted, then they would be simultaneously diagonalizable and a linear combination would exist which would have at least one 0-eigenvalue. Since two-qubit Pauli operators only have four eigenvalues total, having a single $0$ eigenvalue fails the UDA condition.

Furthermore it is easy to show by exhaustive search that there exists no set of more than $5$ mutually anti-commuting Pauli operators. So no fewer than $11$ Paulis could be measured.

To show that this set of $11$ Pauli operators is sufficient to be UDA, we construct a parametrization of all $H\in(\mathcal{S}(\mathbf{A}))^{\perp}$;
\begin{equation}
\label{eq:2qubit-proof}
H=\alpha_{1}XX+\alpha_{2}XY+\alpha_{3}XZ+\alpha_{4}YI+\alpha_{5}ZI
\end{equation}
and show that either $H$ has two positive and two negative eigenvalues or $H= 0$. Note that $H$ then has the following form:
\begin{eqnarray*}
\begin{bmatrix}
\alpha_5 & 0 & \alpha_3+\alpha_4i & \alpha_1+\alpha_2i\\
0 & \alpha_5 & \alpha_1-\alpha_2i & -\alpha_3+\alpha_4i\\
\alpha_3-\alpha_4i & \alpha_1+\alpha_2i & -\alpha_5 & 0\\
\alpha_1-\alpha_2i & -\alpha_3-\alpha_4i & 0 & -\alpha_5\\
\end{bmatrix}.
\end{eqnarray*}

The determinant of $H$ can be calculated and the result is:
\begin{eqnarray*}
&\alpha_5^4+\alpha_5^2\vert\alpha_3+\alpha_2i\vert^2+\alpha_5^2\vert\alpha_1+\alpha_2i\vert^2 \nonumber \\
+&\vert\alpha_3-\alpha_4i\vert^4+\vert\alpha_3-\alpha_4i\vert^2\vert\alpha_1+\alpha_2i\vert^2+\vert\alpha_3-\alpha_2i\vert^2\alpha_5^2 \nonumber \\
+&\vert\alpha_1-\alpha_2i\vert^4+\vert\alpha_1-\alpha_2i\vert^2\vert\alpha_3-\alpha_4i\vert^2+\vert\alpha_1-\alpha_2i\vert^2\alpha_5^2.
\end{eqnarray*}

This quantity, being the sum of non-negative terms, is greater than or equal to 0. Equality is reached if and only if all terms in the sum are 0, which only occurs when $\alpha_1=\alpha_2=\alpha_3=\alpha_4=\alpha_5=0$. Since $H$ is a $4$-by-$4$ traceless Hermitian matrix, it can only have positive determinant if and only if it has exactly two positive and two negative eigenvalues.

The same logic follows for any set that is unitarily equivalent to this set. A particular class of unitary operators which maps the set of Pauli operators to itself is called the Clifford group. Thus, the set $\mathbf{A}$ and any set which is Clifford equivalent to it are our optimum sets of Pauli measurement operators for two-qubit pure-state tomography.
\end{proof}

\subsection{Three-qubit system}

The situation for the $3$-qubit case is much more complicated.
We start by noticing that
\begin{eqnarray}
V&=&IIZ+IZI+ZII+ZZZ\nonumber\\
&=&4\left(|000\rangle\langle000|-|111\rangle\langle111|\right)
\end{eqnarray}
has one positive and one negative eigenvalue. Therefore, if  the set $\mathbf{F}_{1}=\{IIZ,IZI,ZII,ZZZ\}$ is a subset of $\mathcal{S}(\mathbf{A})^{\perp}$, the set $\mathbf{A}$ cannot UDA all pure states.
Similarly any set $\mathbf{F}_{i}$ which is Clifford equivalent to $\mathbf{F}_{1}$ cannot be a subset of $\mathcal{S}(\mathbf{A})^{\perp}$. Sets such as these we call failing sets.

\begin{definition}
A failing set $\mathbf{F}$ is a set of Pauli operators such that there exists a nonzero real combination of elements chosen from $\mathbf{F}$ such that it has only $1$ positive eigenvalue or $1$ negative eigenvalue.
\end{definition}

Namely, for an arbitrary pure state $|\phi\rangle$ to
be UDA by measuring operators in a set $\mathbf{A}$, $\mbox{span(}\mathbf{F}_{i})\not\subset(\mbox{span}(\mathbf{A}))^{\perp}$
holds for every set $\mathbf{F}_{i}$ that is Clifford equivalent to $\mathbf{F}_{1}$.
Thus, for all 945 sets of $\mathbf{F}_{i}$, at least one element in each $\mathbf{F}_{i}$
should be included in $\mbox{span}(\mathbf{A})$.

\begin{theorem}
\label{th:3qubits}
The following set of $31$ Pauli operators are sufficient to UDA
any given three-qubit pure state $|\phi\rangle$
\begin{eqnarray}
\label{eq:3qubit}
\mathbf{A} & = & \{IIX,IIY,IIZ,IXI,IXX,IXY,IYI,IYX,\nonumber \\
 &  & IYY,IZI,XIZ,XXX,XXY,XYX,XYY,\nonumber \\
 &  & XZX,XZY,YXX,YXY,YXZ,YYX,YYY,\nonumber \\
 &  & YYZ,YZI,ZII,ZXZ,ZYZ,ZZX,ZZY,\nonumber \\
 &  & ZZZ,III\}, \label{pauli3}
\end{eqnarray}
and no set with less than $31$ Pauli operators can be UDA for
all three-qubit pure states.
\end{theorem}

Similarly to the two-qubit case this set is obtained by finding the largest set of Pauli operators which do not contain any of the identified failing sets and taking the complement producing the smallest set of measurement operators which could UDA all pure states.

To show that this set $\mathbf{A}$  will be UDA for any pure state, we look at the traceless Hermitian operator $H\in(\mbox{span}(\mathbf{A}))^{\perp}$,
where

\begin{eqnarray*}
H & = & \alpha_{1}IXZ+\alpha_{2}IYZ+\alpha_{3}IZX+\alpha_{4}IZY+\alpha_{5}IZZ\nonumber \\
 &  & +\alpha_{6}XII+\alpha_{7}XIX+\alpha_{8}XIY+\alpha_{9}XXI\nonumber \\
 &  & +\alpha_{10}XXZ+\alpha_{11}XYI+\alpha_{12}XYZ+\alpha_{13}XZI\nonumber \\
 &  & +\alpha_{14}XZZ+\alpha_{15}YII+\alpha_{16}YIX+\alpha_{17}YIY\nonumber \\
 &  & +\alpha_{18}YIZ+\alpha_{19}YXI+\alpha_{20}YYI+\alpha_{21}YZX\nonumber \\
 &  & +\alpha_{22}YZY+\alpha_{23}YZZ+\alpha_{24}ZIX+\alpha_{25}ZIY\nonumber \\
 &  & +\alpha_{26}ZIZ+\alpha_{27}ZXI+\alpha_{28}ZXX+\alpha_{29}ZXY\nonumber \\
 &  & +\alpha_{30}ZYI+\alpha_{31}ZYX+\alpha_{32}ZYY+\alpha_{33}ZZI.
\end{eqnarray*}

It can be shown that $H$ either has at least two positive and two
negative eigenvalues or $H=0$ (see Appendix for details). Therefore, set $\mathbf{A}$
and any set which is Clifford equivalent to it are our optimum Pauli measurement sets for $3$-qubit pure-state tomography.

\section{Stability of the protocol against depolarizing noise}

Before we test our protocol experimentally, we would like to understand how robust it is given states that are not pure. Due to noise in the implementation, we often end up with some
mixed state which is close to our ideal pure state. Therefore, for
the protocol to work in practice, one requires it to return a density
matrix with high fidelity with respect to our input state when it has high purity.
We generate a random pure state $|\phi\rangle$ from the Haar measure as
our desired ideal state, then run it through a depolarizing channel
to get a noisy mixed state $\rho=\eta \frac{ I }{d}+(1-\eta)|\phi\rangle\langle\phi|$.
We could then generate all Pauli measurement results $\{\mbox{Tr}(\rho\sigma_{k})=M_{k}\}$,
where $\sigma_{k}$ is Pauli observable of given dimension. Pick results
determined by our optimum Pauli measurement set, we run a maximum likelihood
estimation to get a density matrix reconstruction. Our
protocol is tested over a range of different $\eta$, and the
results are shown in Figure~\ref{Fig1}. We can see that for small noise $\eta$,
the simulated state is very close to pure, and the protocol returns a
high-fidelity density matrix reconstruction. As noise $\eta$
increases, the pure state assumption becomes less useful, and our protocol yields
a low fidelity estimation.

\begin{figure}[htb]

\includegraphics[width=0.4\textwidth]{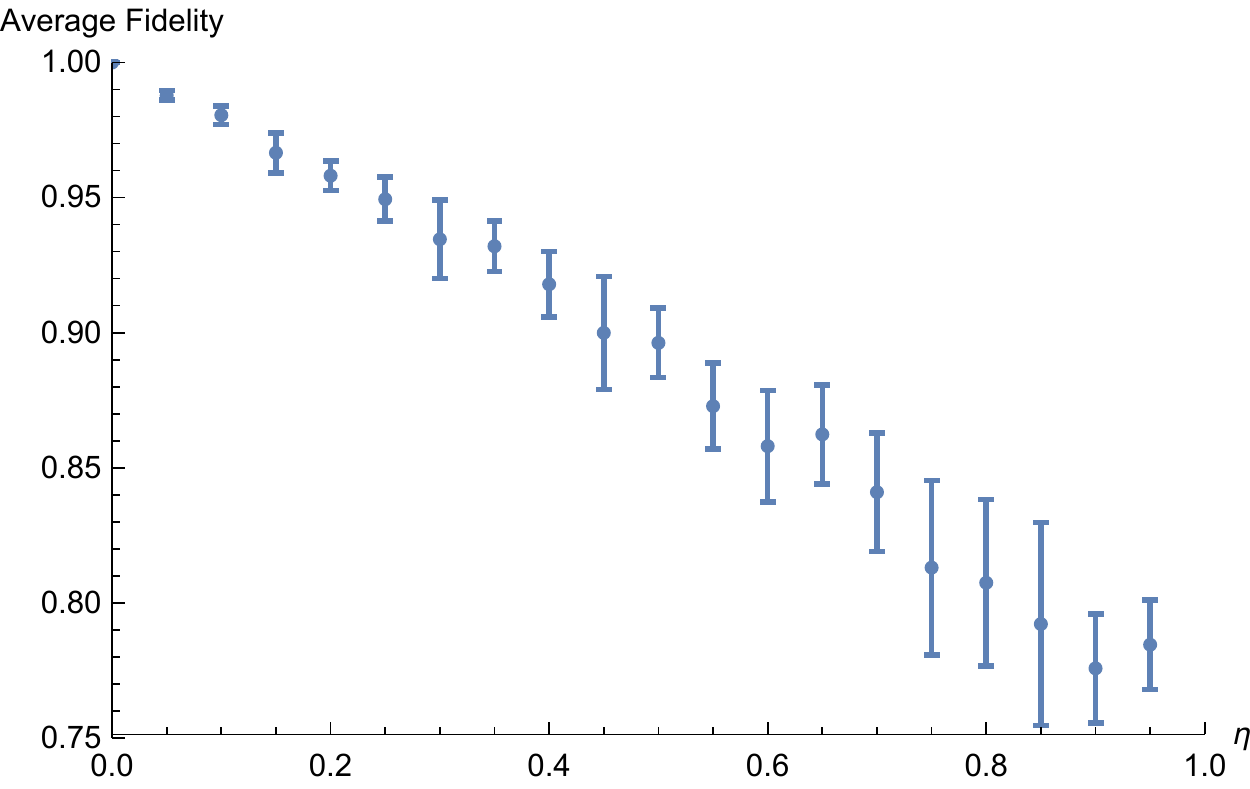}

\caption{The average fidelity of reconstructed density matrices compared to
the ideal state using an optimum Pauli measurement set for 3 qubits. The error bars
are given by standard deviation of the said fidelity over 100 instances. For small noise $\eta$,
the state is very close to pure, and the protocol returns a
high fidelity density matrix reconstruction. As noise $\eta$
increases, the pure state assumption becomes less useful, which yields
a low fidelity estimation. }
\label{Fig1}
\end{figure}

\section{Experiments in NMR
systems}

A nuclear magnetic resonance (NMR) system is an ideal testbed for our protocol. However,
the creation of a pure state in NMR requires unrealistic experimental conditions such as extremely low temperatures or high magnetic fields, which makes it impractical for a liquid sample. To overcome this problem, one can prepare a pseudo-pure state (PPS) alternatively
\begin{align}
\rho_{\text{PPS}}=\frac{1-\epsilon}{2^N}{\mathbb{I}}+\epsilon|\phi\rangle\langle\phi|,
\end{align}
where $\mathbb{I}$ is the identity matrix and $\epsilon\sim10^{-5}$ represents the polarization. For a traceless Pauli observable $\sigma$, only the pure state portion $\epsilon|\phi\rangle\langle\phi|$ contributes to the measurement result. Therefore, the behavior of a system in the PPS is exactly the same as it would be in the pure state.

To test our protocol, we carried out the experiments in $2$- and $3$-qubit NMR quantum
systems, respectively. The qubits in the $2$-qubit system are denoted by the $^{13}$C and $^{1}$H spins of $^{13}$C-labeled Chloroform diluted in acetone-d6 on a Bruker DRX-500 MHz spectrometer, and in the $3$-qubit system by the $^{13}$C, $^{1}$H and $^{19}$F spins in Diethyl-fluoromalonate dissolved in d-chloroform on a Bruker DRX-400 MHz spectrometer. The molecular structures and relevant parameters are shown in Fig. \ref{fig:molecule}, and the corresponding natural Hamiltonian for each system can be described as
\begin{align}\label{Hamiltonian}
\mathcal{H}_{int} = \sum\limits_{i = 1} {\pi\nu_i \sigma _z^i}  + \sum\limits_{i < j, = 1} {\frac{{\pi {J_{ij}}}}{2}\sigma _z^i\sigma _z^j},
\end{align}
where $\nu_i$  is the resonance frequency of spin $i$ and $J_{ij}$ are the scalar coupling constants between spins $i$ and $j$. All parameters are listed in the right table of Fig. \ref{fig:molecule}. Note that in experiment we set $\nu_i=0$ in the multi-rotating frame for simplicity.

In experiment, the entire tomography process for a PPS becomes: given measurements $\mbox{Tr}(\rho\sigma_{k})= \epsilon\mbox{Tr}(\rho_{t}\sigma_{k})=M_{k}$, find a density matrix $\rho_{rec}$ to best fit the data $M_k$. In order to evaluate the performance of our protocol, two comparisons will be made. First, we compare the reconstructed state using the optimum number of Pauli measurements with the one obtained with full tomography. It gives us an idea how good the reconstruction is, and whether the protocol works. Second, we compare our result with the state reconstructed by randomly choosing Pauli measurements. This tells us how different the performance is between selecting the optimum set and a random set of Pauli measurements.

\subsection{Pure state tomography for a 2-qubit state }

\begin{figure}[htb]%
\begin{center}
\includegraphics[width= .9\columnwidth]{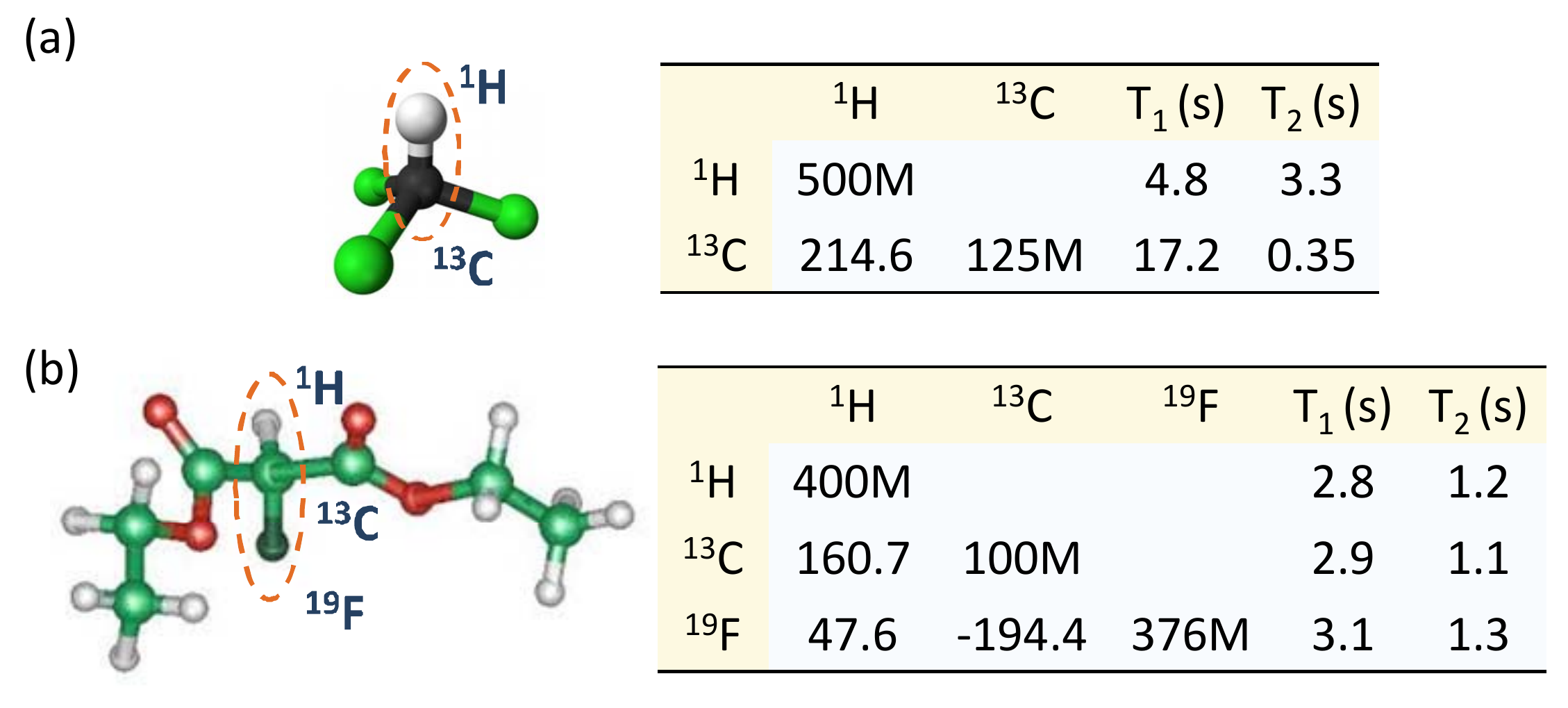}%
\end{center}
\setlength{\abovecaptionskip}{-0.1cm}
\caption{Molecular structure of (a) 2-qubit sample $^{13}$ C-labeled Chloroform and (b) 3-qubit sample Diethyl-fluoromalonate. The corresponding tables on the right side summarize the relevant NMR parameters at room temperature, including the Larmor frequencies (diagonal, in Hertz), the J-coupling constant (off-diagonal, in Hertz) and the relaxation time scales $T_1$ and $T_2$. }\label{fig:molecule}
\end{figure}

For the $2$-qubit protocol, the system is firstly initialized to the PPS
\begin{align}\label{pps2}
\rho_{00}=\frac{1-\epsilon}{4}{\mathbb{I}}+\epsilon\ket{00}\bra{00}
\end{align}
via spatial average technique \cite{cory1997ensemble,lu2015nmr}, and the NMR signal of this PPS is used as references for further comparisons with the tomographic results. We then turn on the transversal field with the strength $\omega_{x}$ (in terms of radius), so in double-rotating frame the Hamiltonian becomes
\begin{align}\label{Hamiltoniantrans}
\mathcal{H}=\frac{\omega_{x}}{2}\left(\sigma_{x}^{1}+\sigma_{x}^{2}\right)+\pi \frac{J_{12}}{2} \sigma_{z}^{1}\sigma_{z}^{2}
\end{align}
By ignoring the identity in $\rho_{00}$, the system should evolve to a time-dependent pure state
\begin{align}\label{pure2}
\ket{\phi}=\alpha (t)|00\rangle+\beta (t) (|01\rangle+|10\rangle)/{\sqrt{2}}+\gamma(t)|11\rangle.
\end{align}
We measured in total 16 different states at a few different time steps and the corresponding Pauli
observables for each state. The reconstructed density matrices for the first and sixteenth
experiments are shown in Figs. \ref{NMRfig1} and \ref{NMRfig2}, respectively. Note that as
the time progresses, the relaxation becomes more prominent, where
the purity of state $\mbox{Tr}(\rho^{2})$ drops. Since our protocol
is designed for pure-state tomography, the performance of our protocol is expected
to drop along with the decrease of purity in a quantum state.

In order to further demonstrate the advantages of our protocol, we
compare it to a quantum state tomography with
Pauli measurements. Using the same number of random Pauli measurements,
one could also perform the maximum likelihood method to get a reconstruction
of the density matrix. Note that the optimum set of $11$ Pauli measurements may
be randomly hit in this case, which means the best performance
of random Pauli measurement algorithm is the same compared with
our protocol. However, in a realistic setting, only one set of
random Pauli measurements will be chosen. To show the advantage of our protocol, we only have to outperform the
average case of this random algorithm.

We randomly generated $11$ distinct 2-qubit Pauli measurements (including identity),
and used the maximum likelihood method to get an estimate of our density matrix.
If the density matrix given by this set of measurements is not unique,
the maximum likelihood method runs multiple times to get an average
estimation. For each experiment, 100 sets of random Pauli measurements
were chosen. The result is shown in Fig. \ref{NMRfig2}. We can see
that for high purity, our method significantly outperforms the random
Pauli algorithm. The advantage decreases as purity decreases, which indicates
our method is more efficient for a state that is close to pure.

\begin{figure}[htb]
\includegraphics[width=0.2\textwidth]{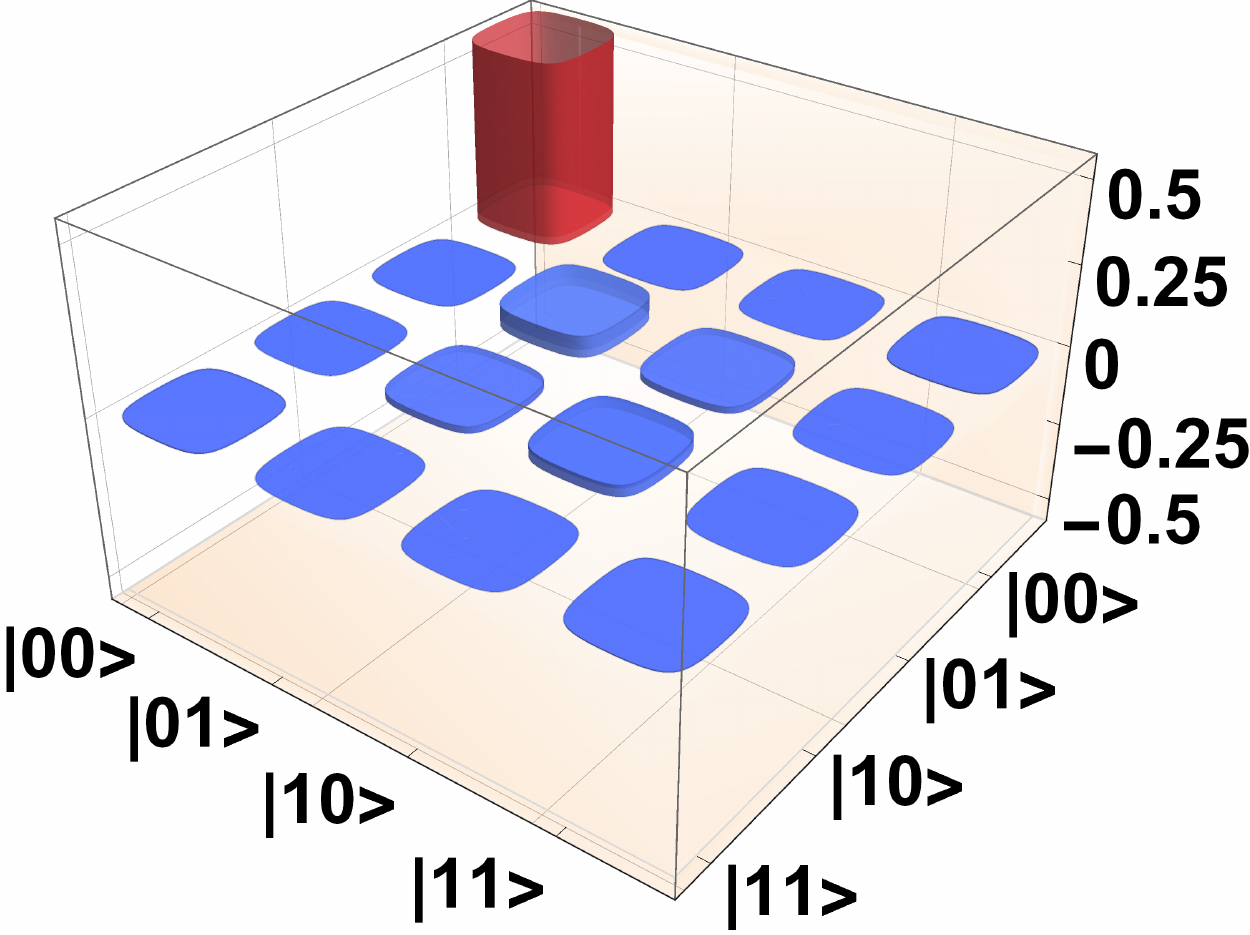}
\includegraphics[width=0.2\textwidth]{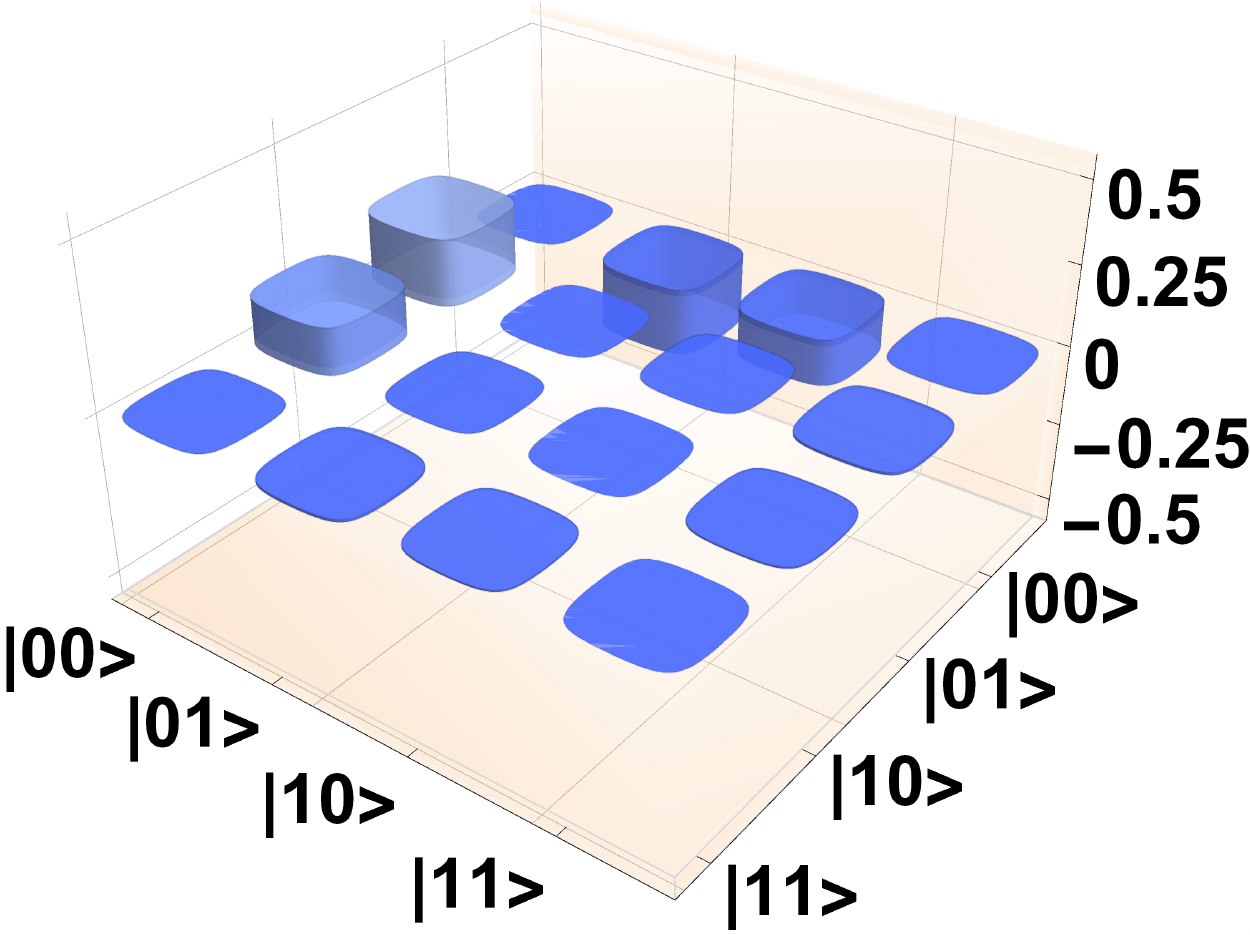}
\includegraphics[width=0.2\textwidth]{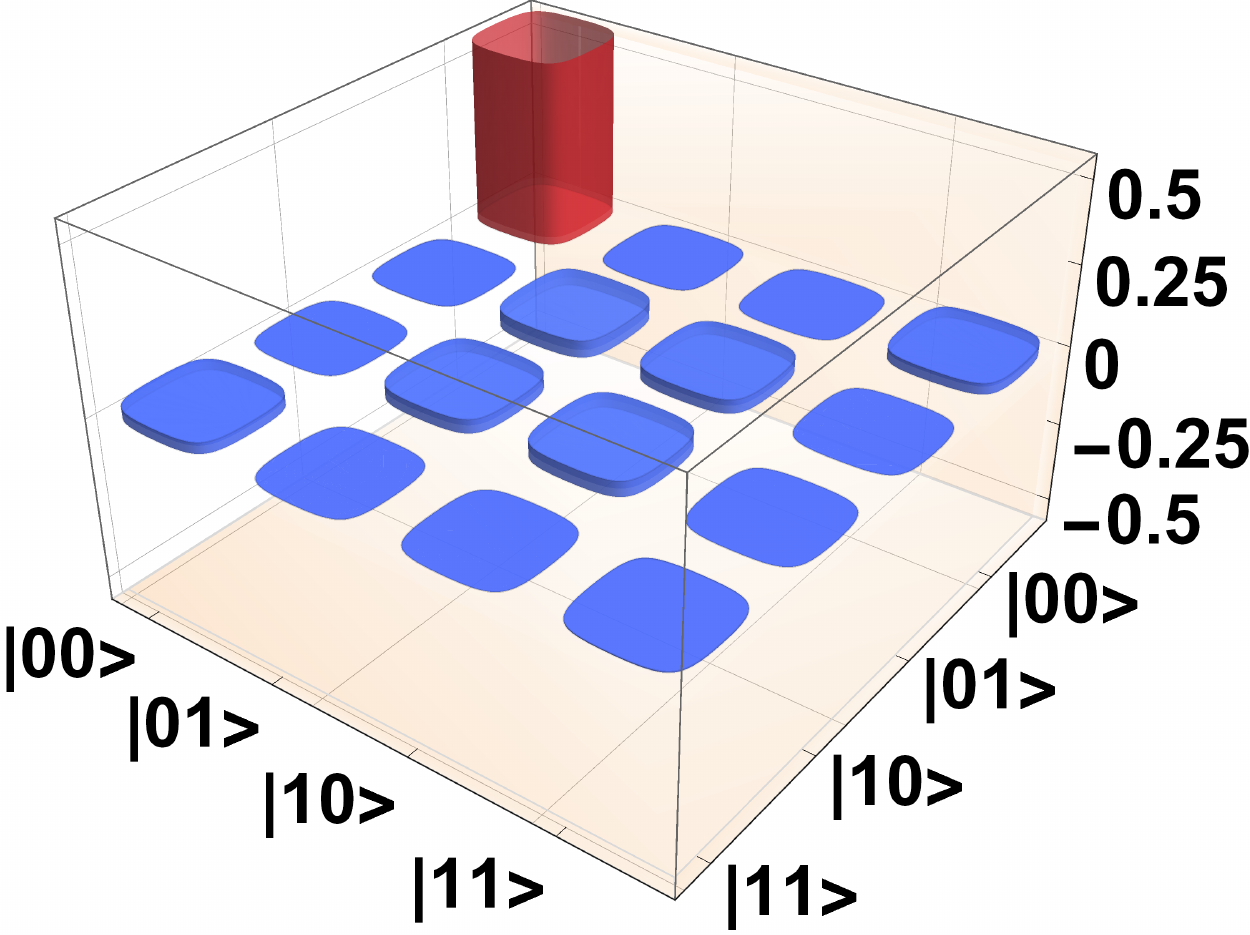}
\includegraphics[width=0.2\textwidth]{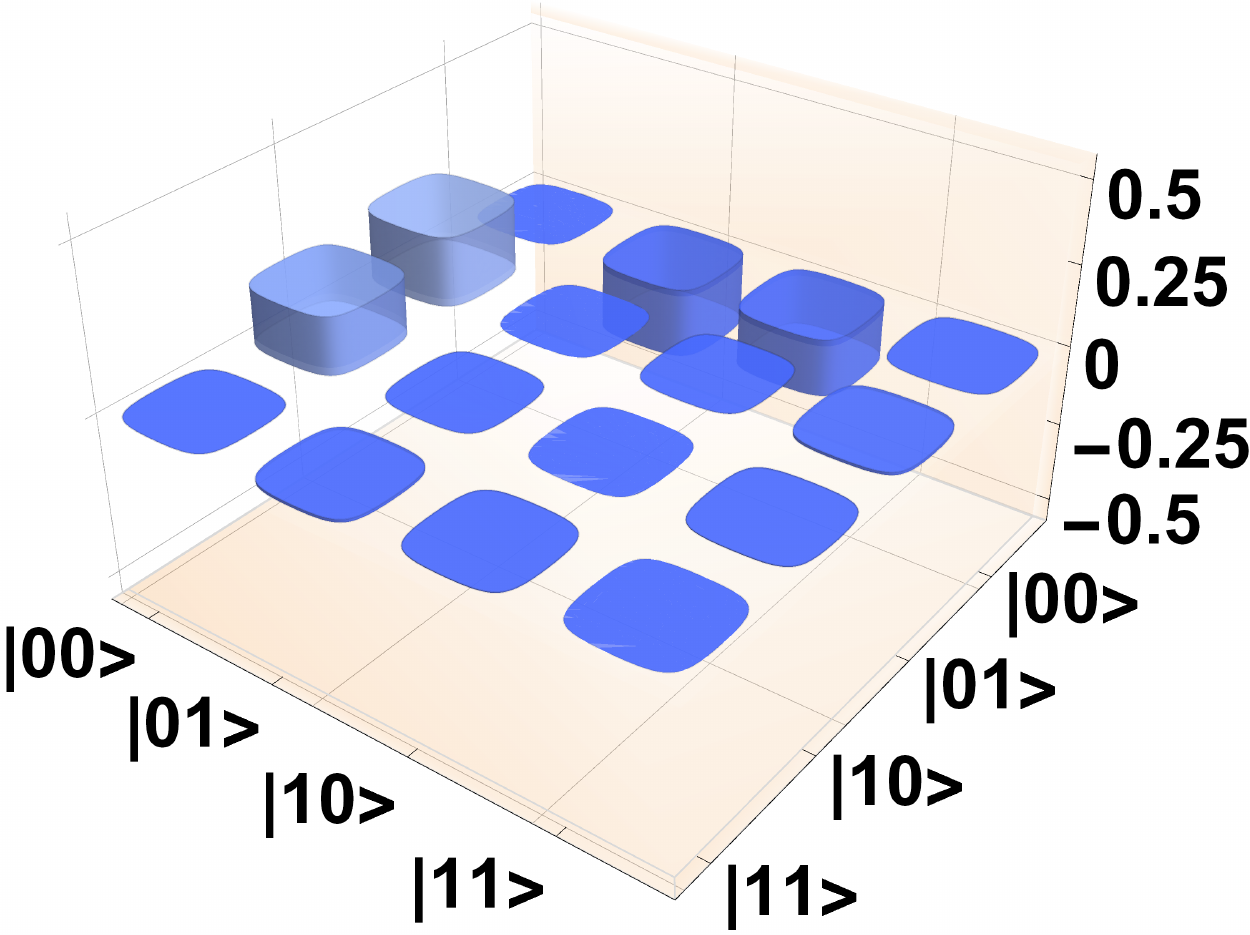}

\caption{The reconstruction of density matrix for state number one. The upper
two figures are real and imaginary part of density matrix of state
reconstruction using all $16$ Pauli measurements. The bottom two figures
are real and imaginary part of density matrix of state reconstruction
using $11$ optimum Pauli measurements described earlier. The fidelity
between the two density matrices is 0.992. }\label{NMRfig1}
\end{figure}

\begin{figure}[htb]
\includegraphics[width=0.4\textwidth]{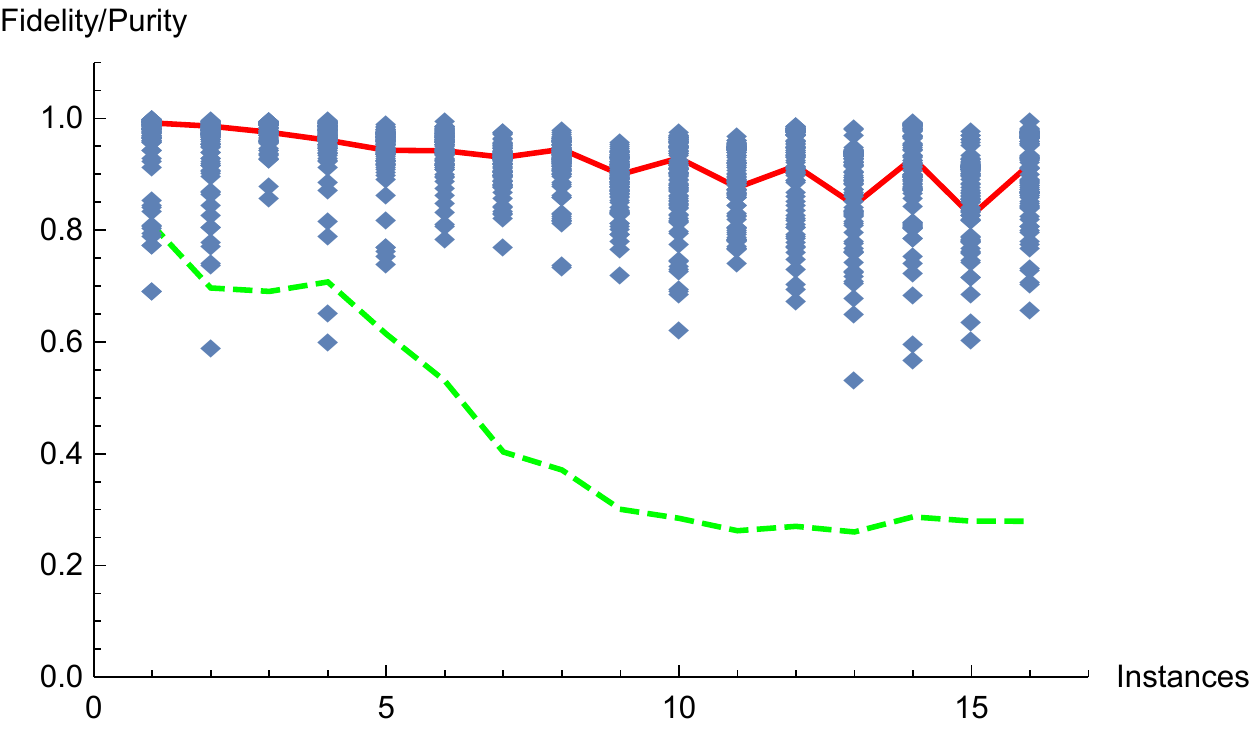}
\caption{Performance of 2-qubit protocol using selected Pauli
measurements against randomly Pauli measurements.The
blue diamond dots are the fidelity between density matrix reconstructed
from all $16$ Pauli measurements with density matrix reconstructed from
random $11$ Pauli measurements. The red line represents fidelity of
reconstruction using our protocol, and the green dashed line shows
the purity of density matrix reconstructed from all Pauli measurements.} \label{NMRfig2}
\end{figure}

\subsection{Pure state tomography for a 3-qubit state }

For $3$-qubit system, we are interested in the GHZ state $|\text{GHZ}\rangle=(|000\rangle+|111\rangle)/\sqrt{2}$. Here, we measured all $64$ Pauli measurements, and only use $31$ of
them described in Eq. \ref{pauli3} for our protocol.
As shown in Fig. \ref{NMRfig3}, only using less than half of the desired measurements, we reconstructed density matrices for the GHZ state with 0.96 fidelity. We then compare it to a quantum state tomography algorithm implementing $31$ random Pauli measurements (including identity). Since the number of unused Pauli measurements
are much more compared to the 2-qubit case, we are less likely to
hit the optimum set in this random algorithm. By implementing a similar
maximum likelihood reconstruction, we found the average fidelity of
this random algorithm to be 0.87 with standard deviation of 0.16.
The detailed result is shown in Fig. \ref{NMRfig4}, which shows clearly that our
protocol has a decent advantage over the average case in the randomized
algorithm.

\begin{figure}[htb]
\includegraphics[width=0.2\textwidth]{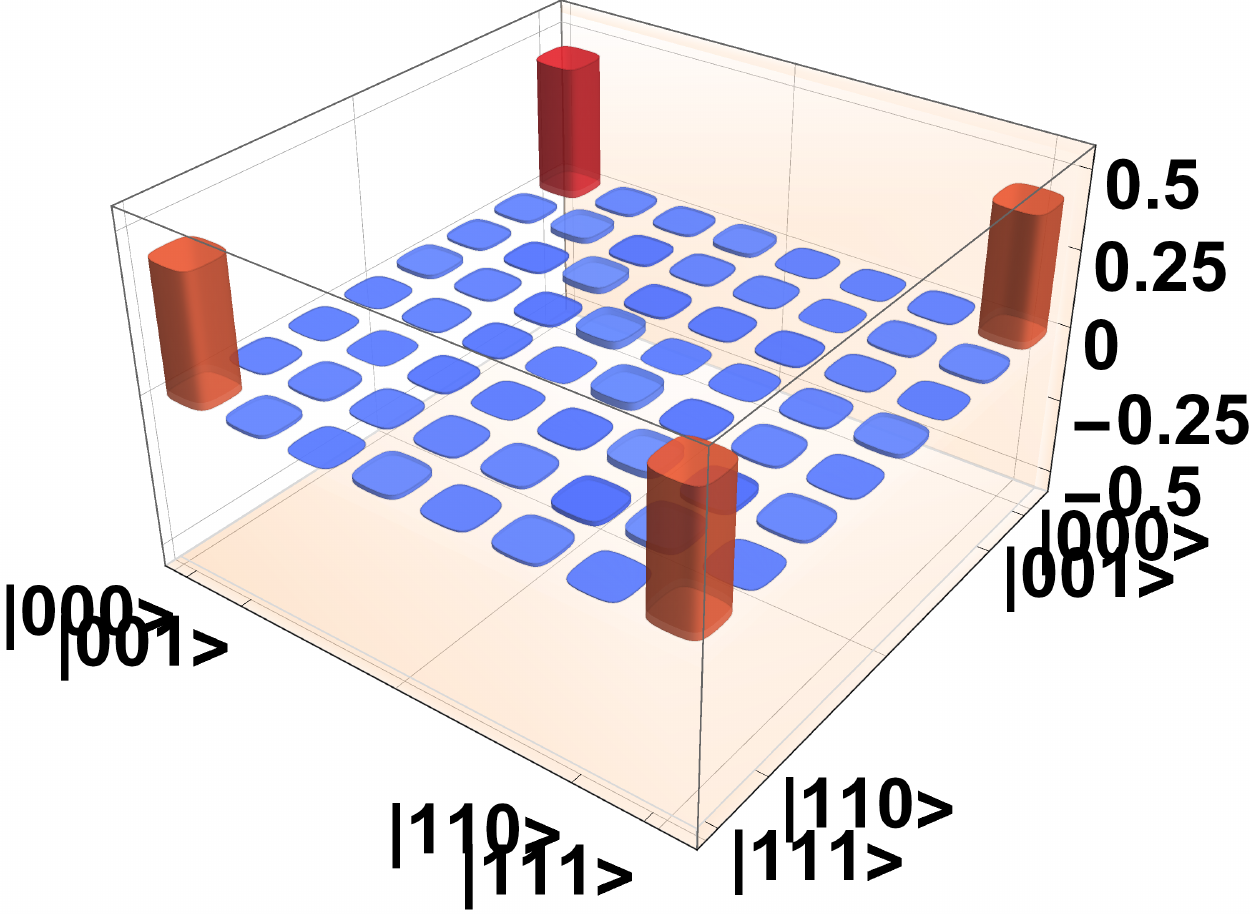}
\includegraphics[width=0.2\textwidth]{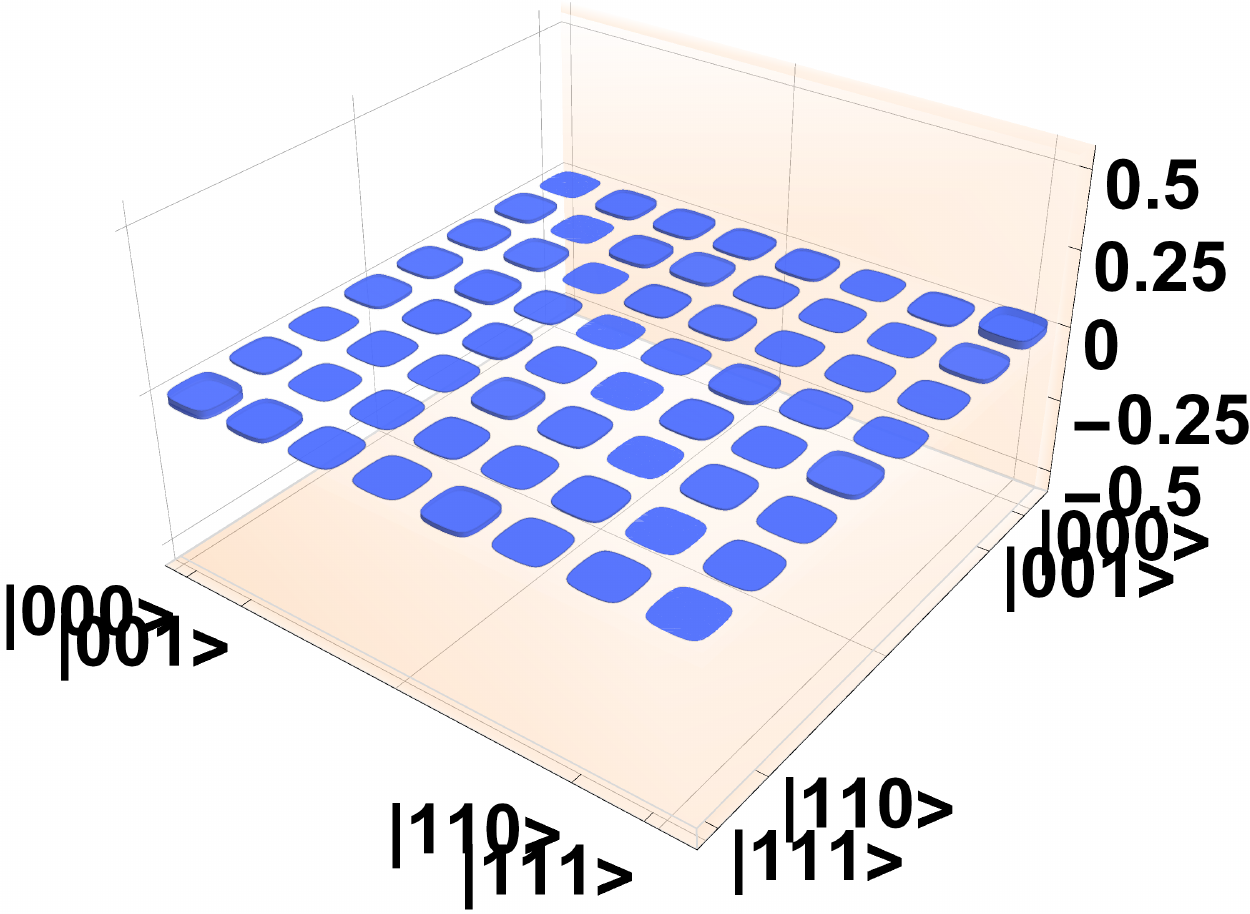}

\includegraphics[width=0.2\textwidth]{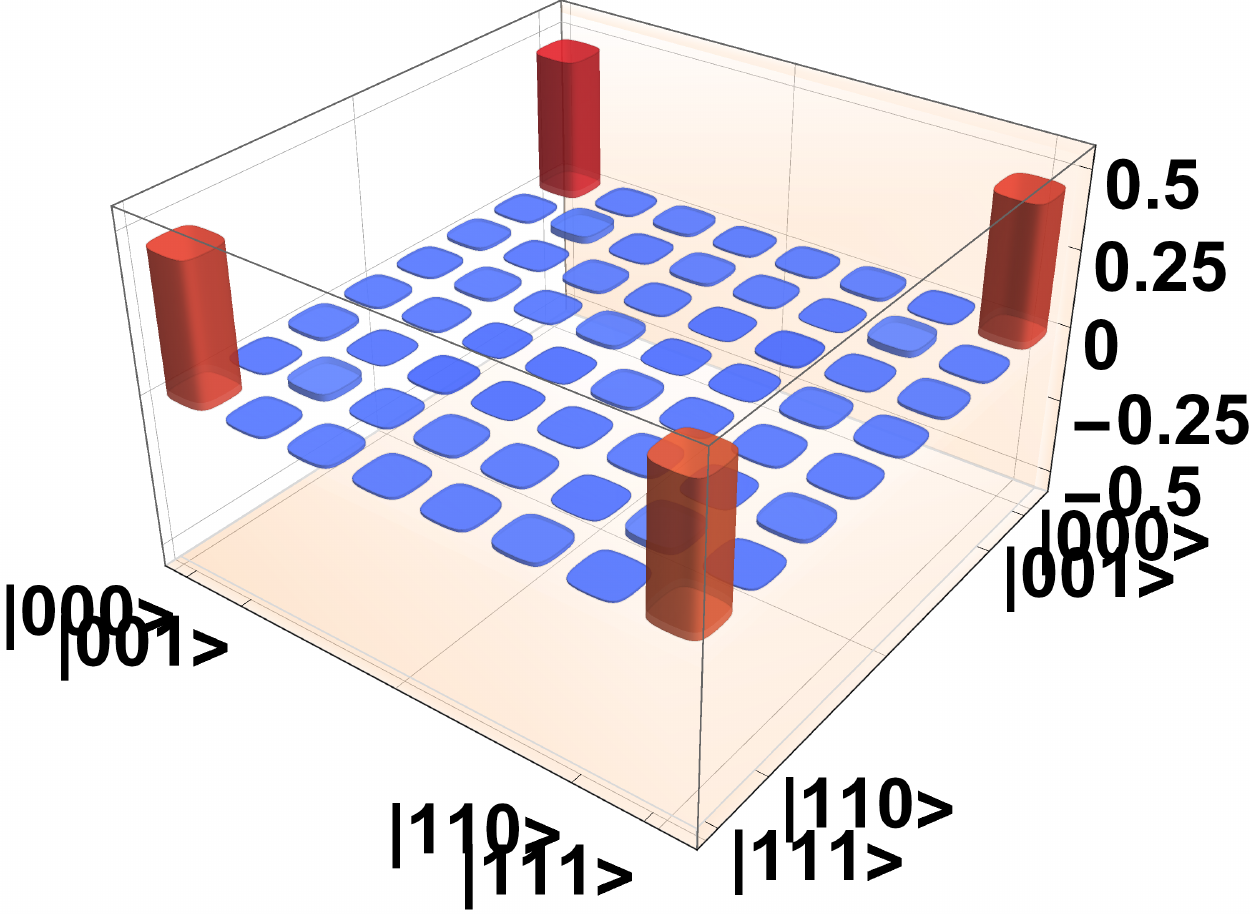}
\includegraphics[width=0.2\textwidth]{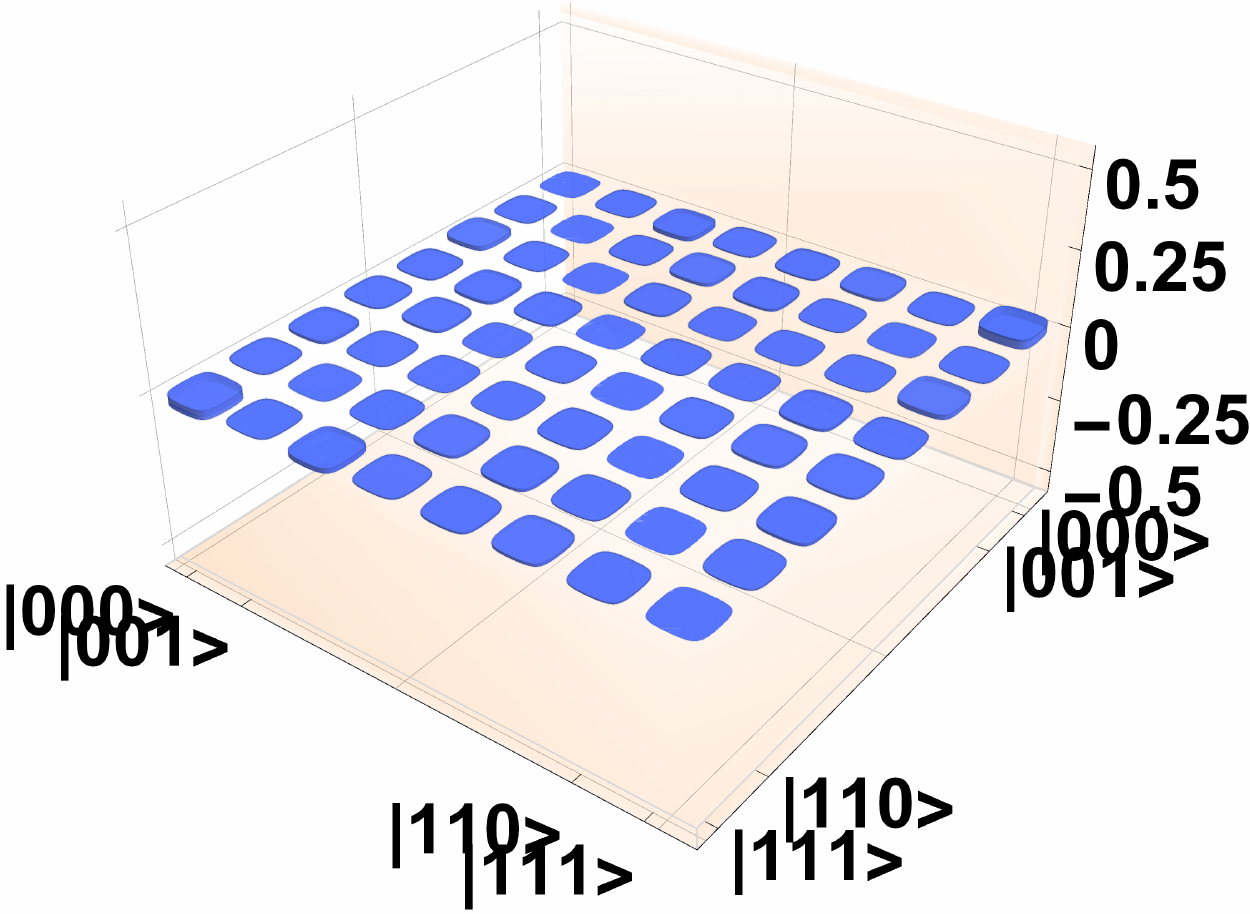}

\caption{The reconstruction of density matrix for GHZ state. The upper two
figures are real and imaginary part of density matrix of state reconstruction
using all $64$ Pauli measurements. The bottom two figures are real and
imaginary part of density matrix of state reconstruction using $31$
optimum Pauli measurements described in Eq.~\ref{eq:3qubit}. The fidelity between
the two density matrices is 0.960.} \label{NMRfig3}
\end{figure}

\begin{figure}[htb]
\includegraphics[width=0.4\textwidth]{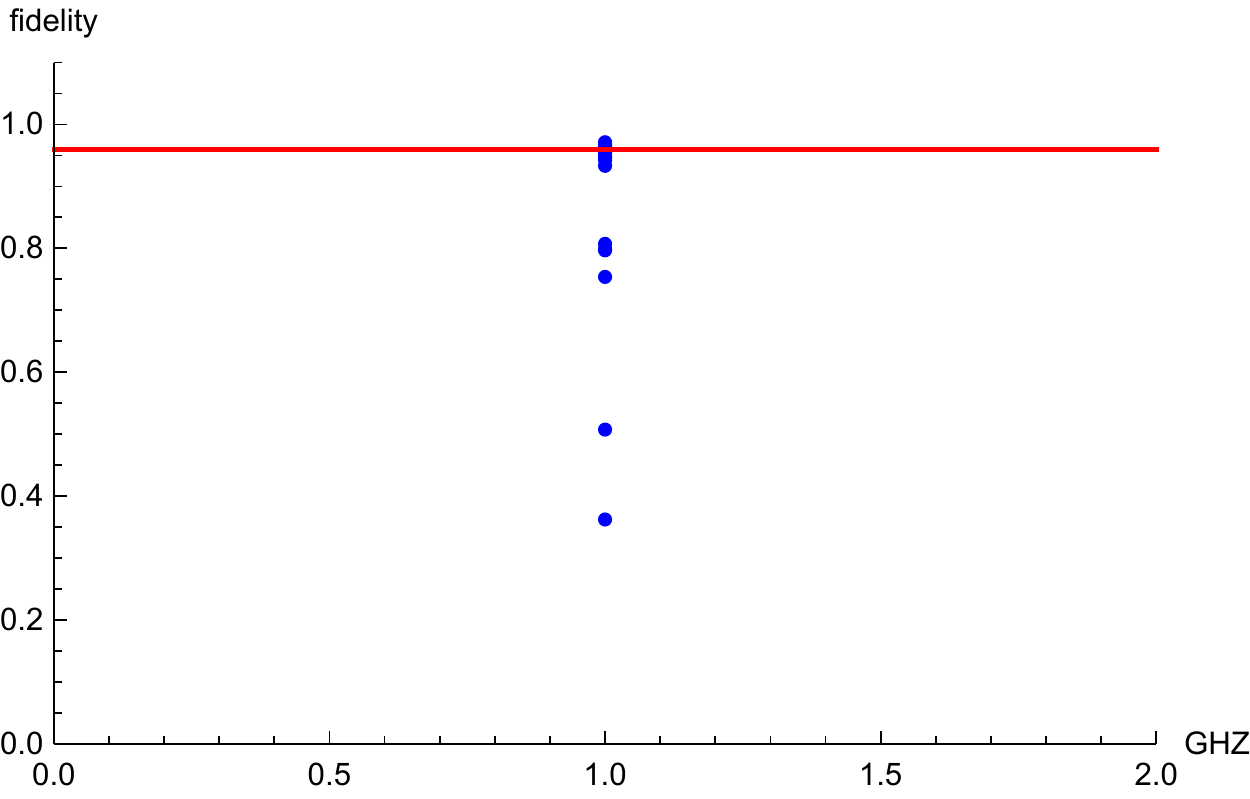}
\caption{Performance of $3$-qubit protocol using selected Pauli
measurements against randomly Pauli measurements. Blue dots represents fidelity between density matrix reconstructed from all $64$ Pauli measurements and density matrix reconstructed from random
$31$ Pauli measurements. The red squre represents fidelity of reconstruction
using our protocol.} \label{NMRfig4}
\end{figure}

\section{Application to tomography in optical systems}

Figure~\ref{opticsmeas} depicts a typical scheme for measuring a polarization-encoded $n$-photon state~\cite{altepeter05,prevedel07,chen07,prevedel11,erven14,hamel14}. Quarter- and half-waveplates in each photon's path are rotated to choose a separable polarization basis. We call the set of angles specifying each waveplate's position the \emph{setting} of the measurement. The $n$-photon state is projected onto the basis set by the waveplate angles with $n$ polarizing beamsplitters. A single-photon detector is present in each of the $2n$ output ports of the beamsplitters, and $n$-fold coincident detections among the $n$ paths are counted. There are $2^n$ combinations of $n$-fold coincident detection events that correspond to a state with one photon entering each of the $n$ beamsplitters before being detected in one of the two output ports. Summing the total number of $n$-fold coincidences over these $2^n$ combinations gives the total number of copies of the state detected by the measurement.

\begin{figure}
\includegraphics{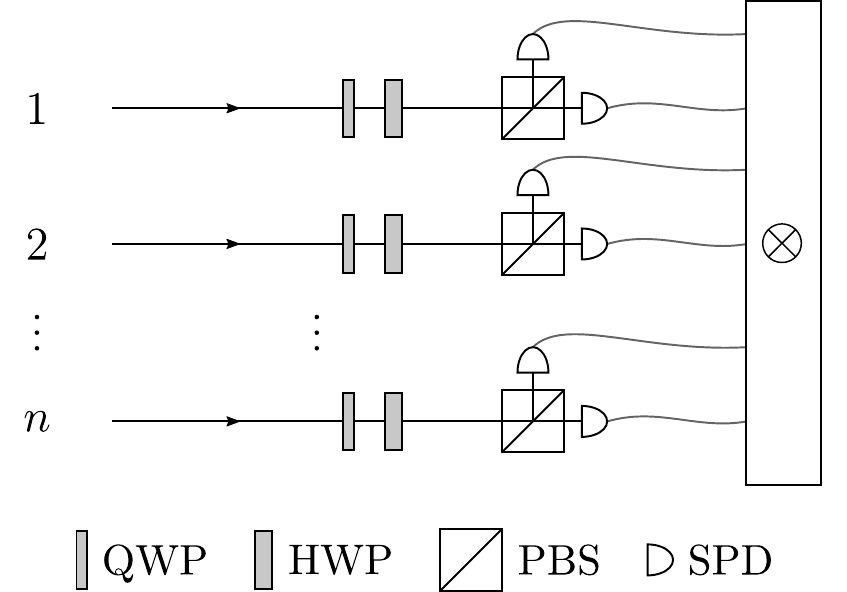}
\caption{Measurement scheme for a polarization-encoded $n$-photon state. The $n$-qubit state is encoded in the polarizations of the $n$ photons. Each photon is measured using a quarter-waveplate (QWP), half-waveplate (HWP) and a polarizing beamplitter (PBS) with a single-photon counting detector (SPD) at each of its output ports. The quarter- and half-waveplates are rotated to choose the measurement basis for each photon. Separable projective measurements are performed by counting coincident detection events between all $n$ photons.} \label{opticsmeas}
\end{figure}

A minimum of $3^n$ measurement settings are required for general state tomography using separable projective measurements~\cite{deBurgh05}. We note that, if one performs nonseparable measurements, then general state tomography can be performed with $2^n+1$ measurement settings~\cite{adamson10}. However, these types of measurements are difficult to perform in practice, so we restrict the discussion here to separable ones.

One can think of each setting as a projective measurement that produces results for multiple Pauli operators simultaneously. For example, consider measuring a 2-photon state with the waveplates set such that a photon in the positive eigenstate of the Pauli $X$ or $Y$ operator will be deterministically transmitted at the first or second beamsplitter, respectively. For simplicity we will call this the $XY$ setting. There are four relevant two-fold coincident detection events, which we denote $N_{tt}$, $N_{tr}$, $N_{rt}$, and $N_{rr}$, and where the first and second subscripts represent which output port (i.e. transmitted or reflected) the first or second photon was detected, respectively. These counts can be summed in specific ways to find expectation values of different Pauli operators. For example the expectation value of $\langle XY\rangle$ is given by $\langle XY\rangle=(N_{tt}-N_{tr}-N_{rt}+N_{rr})/N$, where the total number of copies $N$ is given by $N = N_{tt}+N_{tr}+N_{rt}+N_{rr}$. Similarly, $\langle XI \rangle$ can be found with $\langle XI\rangle=(N_{tt}+N_{tr}-N_{rt}-N_{rr})/N$. In total, the $XY$ setting measures the following four Pauli operators:
\begin{eqnarray*}
XY,XI,IY,II.
\end{eqnarray*}

Based on this observation, we can use the results of Theorem~\ref{th:2qubits} and Theorem~\ref{th:3qubits} to reduce the number of settings to UDA pure states. For the two-qubit case, recall that the $11$ Pauli operators to UDA any pure states are
\begin{eqnarray*}
A=\{II,IX,IY,IZ,XI,YX,YY,YZ,ZX,ZY,ZZ\}.
\end{eqnarray*}

Notice that any of the $6$ Paulis with no $I$ component (the two-qubit correlations) only appear in the setting which measures it. However, looking at the remaining $5$ Paulis, $II$ is included in every setting, $IX$ is included in the $YX$ setting, $IY$ in $YY$, $IZ$ in $YZ$. The only operator which does not appear in the settings of the two-qubit correlations is $XI$, so for the two qubit case, $6+1=7$ settings are required to be sufficient for UDA.

And similar analysis can be done for the three qubit case, with the aid of computer search.
We summarize these results as the corollary below.

\begin{corollary}
\label{cor}
Only $7$ settings
\begin{eqnarray*}
\{XI,YX,YY,YZ,ZX,ZY,ZZ\}.
\end{eqnarray*}
are needed to UDA any two-qubit pure states,
compared with $9$ settings needed for general two-qubit state tomography.
And only $19$ settings
\begin{eqnarray}
\{XXZ, XYZ, XZX, XZY, XZZ, YXX, \nonumber\\YXY, YYX, YYY,
YZX, YZY, YZZ, \nonumber\\ ZXX, ZXY, ZXZ, ZYX, ZYY, ZYZ, ZZX\}
\end{eqnarray}
are needed to UDA any three-qubit pure states,
compared with $27$ settings needed for general three-qubit state tomography.
\end{corollary}

We remark that Corollary $1$ is a direct application of Theorem~\ref{th:2qubits} and Theorem~\ref{th:3qubits}. It is possible for even better results to be obtained by including knowledge of settings in the first optimization. However, proving sufficiency becomes more difficult in these cases.

\section{Conclusion}

In this work, we find the most compact Pauli measurement sets for
pure-state tomography on two- and three-qubit systems. The experiments on
two-qubit and three-qubit NMR systems demonstrated the advantages of using
such protocol. We reduced the required number of measurements by $5$ and $33$ for
two- and three-qubit systems, respectively, without significant drop in fidelity.
As a direct application of this result,
we also showed that our scheme can be used to reduce
the number of settings needed for pure-state tomography in quantum optics systems.

A few questions need to be answered before we scale the test to larger
systems. We are able to find the optimum sets for two and three qubits.
However, the method we used to find those sets can not be easily generalized
to larger systems. It remains open whether one can find a general algorithm to
decide the smallest sets of Pauli operators to UDA any pure state for
a system of $n$ qubits. If such an algorithm exists, we would hope that
the number of measurements required
grows linearly with the Hilbert space dimension of the system.

\section*{Acknowledgement} 
We thank Nengkun Yu for helpful discussions. 
This research was supported in part by the Natural Sciences and Engineering Research Council of Canada (NSERC), Canada Research Chairs, Industry Canada, National Key Basic Research Program (2013CB921800 and 2014CB848700), the National Science Fund for Distinguished Young Scholars (NO. 11425523) and National Natural Science Foundation of China (NO. 11375167). 

\bibliography{Pauli}


\appendix
\label{appendix}
\section{Proof of Theorem~\ref{th:3qubits}}

In order to prove Theorem~\ref{th:3qubits}, it suffices to prove
the following result.

\begin{theorem}
Any Hermitian operator perpendicular to
\begin{eqnarray*}
&&\{IIX,IIY,IIZ,IXI,IXX,IXY,IYI,IYX,IYY, IZI,\nonumber \\
&&XIZ,XXX,XXY,XYX,XYY,XZX,XZY,YXX,\nonumber \\
&&YXY,YXZ,YYX,YYY,YYZ,YZI,ZII,ZXZ,ZYZ,\nonumber \\
&&ZZX,ZZY,ZZZ\}
\end{eqnarray*}
must have at least two positive and two negative eigenvalues.
\end{theorem}

\begin{proof}
The proof proceeds as follows. First construct an $8$-by-$8$ traceless Hermitian matrix $H$ which is perpendicular to all the above Pauli operators. This will be a real linear combination of every Pauli operator that is not being measured. This $H$ is then a general description of any Hermitian matrix in the complement of the span of all measured operators. We will show through a case by case analysis that if we assume $H$ only has one positive eigenvalue, then it follows that $H$ must be the zero matrix. A similar argument holds for having only one negative eigenvalue therefore $H$ must have at least two positive and two negative eigenvalues.

Let us begin by constructing $H$ which is a real linear combination of the 33 Pauli operators not being measured (excluding the identity). $H$ is then:

\begin{eqnarray*}
H&=&x_{1} IXZ + x_2 IYZ + x_3 IZX + x_4 IZY \\
&&+ x_5 IZZ + x_6 XII + x_7 XIX + x_8 XIY\\
&&+ x_9 XXI + x_{10} XXZ + x_{11} XYI + x_{12} XYZ \\
&&+ x_{13} XZI + x_{14} XZZ + x_{15} YII + x_{16} YIX\\
&& + x_{17} YIY + x_{18} YIZ + x_{19} YXI + x_{20} YYI \\
&&+ x_{21} YZX + x_{22} YZY + x_{23} YZZ + x_{24} ZIX \\
&&+ x_{25} ZIY + x_{26} ZIZ + x_{27} ZXI + x_{28} ZXX \\
&&+ x_{29} ZXY + x_{30} ZYI + x_{31} ZYX + x_{32} ZYY\\
&& + x_{33} ZZI.
\end{eqnarray*}

Writing $H$ in matrix form will give the form:
\begin{eqnarray}
\begin{bmatrix}
c_{11} & c_{12}& c_{13}& c_{14}&c_{15} &c_{16} &c_{17} & 0 \\
c_{12}^{\ast} & c_{22} & c_{23} & c_{24}& c_{25}& c_{26}&0 &c_{28}  \\
c_{13}^{\ast} & c_{23}^{\ast} & c_{33}& c_{34}& c_{35}& 0& c_{37} & c_{38}  \\
c_{14}^{\ast} & c_{24}^{\ast} & c_{34}^{\ast}& c_{44}&0 & c_{46}& c_{47}& c_{48} \\
c_{15}^{\ast} & c_{25}^{\ast} & c_{35}^{\ast}& 0 & c_{55}& c_{56}& c_{57}& c_{58}  \\
c_{16}^{\ast} & c_{26}^{\ast} &0 & c_{46}^{\ast}& c_{56}^{\ast}& c_{66}& c_{67}& c_{68}  \\
c_{17}^{\ast} & 0 & c_{37}^{\ast}& c_{47}^{\ast}& c_{57}^{\ast}& c_{67}^{\ast}&c_{77} & c_{78} \\
0 & c_{28}^{\ast} &c_{38}^{\ast} & c_{48}^{\ast}& c_{58}^{\ast}& c_{68}^{\ast}&c_{78}^{\ast} &c_{88}
\end{bmatrix}
\end{eqnarray}
where
\begin{eqnarray*}
c_{11}&=& x_{5}+x_{26} + x_{33} ;\\
c_{22}&=& -x_5-x_{26} + x_{33};\\
c_{33}&=& -x_5+x_{26} - x_{33};\\
c_{44}&=&x_5-x_{26} - x_{33} ;\\
c_{55}&=&x_5-x_{26} - x_{33}=c_{44};\\
c_{66}&=& -x_5+x_{26} - x_{33}=c_{33};\\
c_{77}&=& -x_5-x_{26} + x_{33}=c_{22};\\
c_{88}&=& x_{5}+x_{26} + x_{33}=c_{11};\\
c_{12}&=& x_3+x_{24} - i(x_4+ x_{25});\\
c_{34}&=&-x_3+x_{24}+i(x_4-x_{25});\\
c_{56}&=&x_3-x_{24}-i(x_4-x_{25})=-c_{34};\\
c_{78}&=&-x_3-x_{24}+i(x_4+x_{25})=-c_{12};\\
c_{13}&=&x_1+x_{27} - i(x_2+ x_{30});\\
c_{24}&=&-x_1 +x_{27}+ i(x_2-x_{30});\\
c_{57}&=&x_1-x_{27}-i(x_2-x_{30})=-c_{24};\\
c_{68}&=&-x_1-x_{27}+i(x_2+x_{30})=-c_{13};\\
c_{14}&=&x_{28}-x_{32} - i( x_{29}+x_{31});\\
c_{23}&=&x_{28}+x_{32} +i(x_{29}-x_{31});\\
c_{58}&=&-x_{28}+x_{32}+i(x_{29}+x_{31})=-c_{14};\\
c_{67}&=&-x_{28}-x_{32}-i(x_{29}-x_{31})=-c_{23};\\
c_{15}&=&x_6+x_{13} + x_{14} - i( x_{15} + x_{18} + x_{23});  \\
c_{26}&=&x_6+x_{13}-x_{14}-i(x_{15}-x_{18}-x_{23});\\
c_{37}&=&x_6-x_{13}-x_{14}-i(x_{15}+x_{18}-x_{23});\\
c_{48}&=&x_6-x_{13}+x_{14}-i(x_{15}-x_{18}+x_{23})\\
&&=c_{15}-c_{26}^{\ast}+c_{37}^{\ast};\\
c_{16}&=&x_7- x_{17}  - x_{22}  - i( x_8+  x_{16}+x_{21}) ;\\
c_{25}&=&x_7+x_{17}+x_{22}+i(x_8-x_{16}-x_{21});\\
c_{38}&=&x_7-x_{17}+x_{22}-i(x_8+x_{16}-x_{21});\\
c_{47}&=&x_7+x_{17}-x_{22}+i(x_8-x_{16}+x_{21})\\
&&=c_{16}^{\ast}+c_{25}-c_{38}^{\ast};\\
\end{eqnarray*}
\begin{eqnarray*}
c_{17}&=& x_9+x_{10} - x_{20}-i( x_{11}+ x_{12}+ x_{19});\\
c_{28}&=&x_9-x_{10}-x_{20}-i(x_{11}-x_{12}+x_{19});\\
c_{35}&=&x_9+x_{10}+x_{20}+i(x_{11}+x_{12}-x_{19});\\
c_{46}&=&x_9-x_{10}+x_{20}+i(x_{11}-x_{12}-x_{19})\\
&&=c_{28}^{\ast}+c_{35}-c_{17}^{\ast};\\
\end{eqnarray*}

Note that the main anti-diagonal is all zeros. This was by design, since any set of Pauli operators Clifford equivalent to the result from the hyper-graph dualization program is also a solution, we had the freedom to choose a set which would make the proof simpler. Choosing the set of operators which contained all Pauli operators constructed by tensoring only $X$ operators and $Y$ operators meant $H$ would have zero main anti-diagonal. The only reason for choosing this set is it makes this proof a little simpler.

Here we assume $H$ is a Hermitian matrix with only one positive eigenvalue. We first show all diagonal entries of $H$ must be zero. Observe that $c_{55}=c_{44}$, $c_{66}=c_{33}$, $c_{77}=c_{22}$, $c_{88}=c_{11}$. In order for the traceless condition on $H$ to hold, it is then clear that $c_{11}+c_{22}+c_{33}+c_{44}=0$. If $H$ has some nonzero diagonal entry, then at least one of $c_{11}, c_{22}, c_{33}$ and $c_{44}$ will be positive. Without loss of generality, let $c_{11}>0$, then the submatrix of $H$ formed by the rows $(1, 8)$ and columns $(1,8)$, which will be of the form $c_{11}*I$, will have two positive eigenvalues.
\begin{lemma}
Cauchy's Interlacing Theorem states\cite{horn2012matrix}:\\
Let:
\begin{eqnarray*}
A=
\begin{bmatrix}
B & C\\
C^\dagger & D
\end{bmatrix}
\end{eqnarray*}
be an n-by-n Hermitian matrix, where $B$ has size m-by-m (m$<$n). If the eigenvalues of $A$ and $B$ are $\alpha_1\leq \ldots \leq \alpha_n$ and $\beta_1 \leq \ldots \leq \beta_m$ respectfully. Then:
\begin{eqnarray*}
\alpha_k \leq \beta_k \leq \alpha_{k+n-m}, k=1,\ldots ,m.
\end{eqnarray*}
\end{lemma}
It follows from Cauchy's interlacing property that if a principle submatrix of $H$ has 2 positive eigenvalues then $H$ also has at least two positive eigenvalues.

Hence,  $H$ must be in the following form:

\begin{eqnarray}\label{eq:form}
\resizebox{\linewidth}{!}{$
H=\begin{bmatrix}
0 & c_{12}& c_{13}& c_{14}&c_{15} &c_{16} &c_{17} & 0 \\
c_{12}^{\ast} & 0 & c_{23} & c_{24}& c_{25}& c_{26}&0 &c_{28}  \\
c_{13}^{\ast} & c_{23}^{\ast} & 0& c_{34}& c_{35}& 0& c_{37} & c_{38}  \\
c_{14}^{\ast} & c_{24}^{\ast} & c_{34}^{\ast}& 0&0 & c_{28}^{\ast}+c_{35}-c_{17}^{\ast}& c_{16}^{\ast}+c_{25}-c_{38}^{\ast}& c_{15}-c_{26}^{\ast}+c_{37}^{\ast} \\
c_{15}^{\ast} & c_{25}^{\ast} & c_{35}^{\ast}& 0 & 0& -c_{34}& -c_{24}& -c_{14}  \\
c_{16}^{\ast} & c_{26}^{\ast} &0 &c_{28}+c_{35}^{\ast}-c_{17}& -c_{34}^{\ast}& 0& -c_{23}& -c_{13}  \\
c_{17}^{\ast} & 0 & c_{37}^{\ast}& c_{16}+c_{25}^{\ast}-c_{38}& -c_{24}^{\ast}& -c_{23}^{\ast}&0 & -c_{12} \\
0 & c_{28}^{\ast} &c_{38}^{\ast} & c_{15}^{\ast}-c_{26}+c_{37}& -c_{14}^{\ast}& -c_{13}^{\ast}&-c_{12}^{\ast} &0
\end{bmatrix}.
$}\nonumber
\end{eqnarray}

In fact, under the assumption that $H$ has only $1$ positive eigenvalue, it follows from Cauchy's interlacing theorem that any principle submatrix of $H$ cannot have more than one positive eigenvalue. Otherwise, we will have a contradiction.

Let us look at the submatrix formed by rows $1,2,4,5$ and the same columns. It is a traceless Hermitian matrix with determinant $\vert c_{14}c_{25}-c_{15}c_{24}\vert^2$. Again, if the submatrix has positive determinant, then it must have exactly two positive eigenvalues. Once again by applying Cauchy's interlacing property, $H$ will have at least two positive eigenvalues. This immediately contradictions our assumption. The above argument  implies that, under our assumption $H$ has only $1$ positive eigenvalue, we have $\vert c_{14}c_{25}-c_{15}c_{24}\vert^2\leq 0$. It is not surprising that the inequality holds if and only if the equality holds. Then we have $c_{14}c_{25}-c_{15}c_{24}=0$.

Similarly, by considering other $4$-by-$4$ submatrices constructed from the rows and columns $a,b,4,5$ where $a,b$ are any two of the remain six rows, we can show that:
\begin{eqnarray*}
c_{14}c_{35}-c_{15}c_{34}&=&0;\\
-c_{14}c_{34}^{\ast}-c_{15}(c_{28}+c_{35}^{\ast}-c_{17})&=&0;\\
-c_{14}c_{24}^{\ast}-c_{15}(c_{16}+c_{25}^{\ast}-c_{38})&=&0;\\
-c_{14}c_{14}^{\ast}-c_{15}(c_{15}^{\ast}-c_{26}+c_{37})&=&0;\\
c_{24}c_{35}-c_{25}c_{34}&=&0;\\
-c_{24}c_{34}^{\ast}-c_{25}(c_{28}+c_{35}^{\ast}-c_{17})&=&0;\\
-c_{24}c_{24}^{\ast}-c_{25}(c_{16}+c_{25}^{\ast}-c_{38})&=&0;\\
-c_{24}c_{14}^{\ast}-c_{25}(c_{15}^{\ast}-c_{26}+c_{37})&=&0;\\
-c_{34}c_{34}^{\ast}-c_{35}(c_{28}+c_{35}^{\ast}-c_{17})&=&0;\\
-c_{34}c_{24}^{\ast}-c_{35}(c_{16}+c_{25}^{\ast}-c_{38})&=&0;\\
-c_{34}c_{14}^{\ast}-c_{35}(c_{15}^{\ast}-c_{26}+c_{37})&=&0;\\
-c_{24}^{\ast}(c_{28}+c_{35}^{\ast}-c_{17})+c_{34}^{\ast}(c_{16}+c_{25}^{\ast}-c_{38})&=&0;\\
-c_{14}^{\ast}(c_{28}+c_{35}^{\ast}-c_{17})+c_{34}^{\ast}(c_{15}^{\ast}-c_{26}+c_{37})&=&0;\\
-c_{14}^{\ast}(c_{16}+c_{25}^{\ast}-c_{38})+c_{24}^{\ast}(c_{15}^{\ast}-c_{26}+c_{37})&=&0.
\end{eqnarray*}

The above equations will imply that the $8$-by-$2$ submatrix formed by the $4$-th and $5$-th columns has rank at most $1$.

The same argument can be used to prove that the $8$-by-$2$ submatrices formed by columns $(1,8)$, $(2,7)$ or $(3,6)$ also have rank at most $1$.

As a straightforward consequence, $H$ has rank no more than $4$.

In other words, the $k$-th column and the $(9-k)$-th column are linearly dependant. This means that there exist $\lambda_1,\lambda_2,\lambda_3,\lambda_4$ such that the following equations hold:

\begin{eqnarray}\label{eq:lambda}
&\lambda_1\overrightarrow{C_1}+(1-\lambda_1)\overrightarrow{C_8}=\lambda_2\overrightarrow{C_2}+(1-\lambda_2)\overrightarrow{C_7}=0\\
&\lambda_3\overrightarrow{C_3}+(1-\lambda_3)\overrightarrow{C_6}=\lambda_4\overrightarrow{C_4}+(1-\lambda_4)\overrightarrow{C_5}=0
\end{eqnarray}
Here we have used $\overrightarrow{C_k}$ to represent the $k$-th column of the matrix (\ref{eq:form}).

Let us start with a special case.
Let $\lambda_1=0$. Then $c_{12}=c_{13}=c_{14}=c_{28}=c_{38}=0$ and $c_{15}=c_{26}^{\ast}-c_{37}^{\ast}$. $H$ can be simplified as the following:

\begin{eqnarray*}
\resizebox{\linewidth}{!}{$
H=\begin{bmatrix}
0 & 0& 0& 0&c_{26}^{\ast}-c_{37}^{\ast} &c_{16} &c_{17} & 0 \\
0 & 0 & c_{23} & c_{24}& c_{25}& c_{26}&0 &0  \\
0 & c_{23}^{\ast} & 0& c_{34}& c_{35}& 0& c_{37} & 0  \\
0 & c_{24}^{\ast} & c_{34}^{\ast}& 0&0 & c_{35}-c_{17}^{\ast}& c_{16}^{\ast}+c_{25}& 0 \\
c_{26}-c_{37} & c_{25}^{\ast} & c_{35}^{\ast}& 0 & 0& -c_{34}& -c_{24}&0  \\
c_{16}^{\ast} & c_{26}^{\ast} &0 &c_{35}^{\ast}-c_{17}& -c_{34}^{\ast}& 0& -c_{23}& 0  \\
c_{17}^{\ast} & 0 & c_{37}^{\ast}& c_{16}+c_{25}^{\ast}& -c_{24}^{\ast}& -c_{23}^{\ast}&0 & 0 \\
0 & 0 &0 & 0&0& 0&0 &0
\end{bmatrix}.
$}
\end{eqnarray*}

If we set $c_{23}=c_{24}=c_{34}=0$, then the top-left $4$-by-$4$ submatrix is zero. In this case, the characteristic polynomial of $H$ contains only even powers. Thus $H$ has only one positive eigenvalue implies $H$ has only one negative eigenvalue too. As a consequence, the top-right $4$-by-$4$ submatrix of $H$ has rank exactly $1$.

As a result, any $2$-by-$2$ submatrix of the top-right submatrix must have determinant zero. From suitable choices of submatrices we can obtain the following equations:
\begin{eqnarray}
c_{26}c_{37} & =&0;\\
c_{26}(c_{26}^{\ast}-c_{37}^{\ast})& =&c_{16}c_{25};\\
c_{37}(c_{26}^{\ast}-c_{37}^{\ast})& =&c_{17}c_{35};\\
c_{16}(c_{16}^{\ast}+c_{25})+c_{17}(c_{17}^{\ast}-c_{35}) & =&0.
\end{eqnarray}

 Using the above equations we can obtain:
\begin{eqnarray}
0&=&c_{16}(c_{16}^{\ast}+c_{25})+c_{17}(c_{17}^{\ast}-c_{35})\nonumber \\
&=&c_{16}c_{25}-c_{17}c_{35}+\vert c_{17}\vert^2+\vert c_{16}\vert^2\nonumber \\
&=&c_{26}(c_{26}^{\ast}-c_{37}^{\ast})-c_{37}(c_{26}^{\ast}-c_{37}^{\ast})+\vert c_{17}\vert^2+\vert c_{16}\vert^2\nonumber \\
&=&\vert c_{26}-c_{37}\vert^2+\vert c_{17}\vert^2+\vert c_{16}\vert^2
\end{eqnarray}
This implies $c_{16}=c_{17}=0$ and $c_{26}=c_{37}$. Also since $c_{26}c_{37}=0$ we know that $c_{26}=c_{37}=0$. Furthermore $c_{25}(c_{16}^{\ast}+c_{25})=0$ and $c_{35}(c_{35}-c_{17}^{\ast})=0$ will guarantee $c_{25}=c_{35}=0$. Therefore $H$ is once again the zero matrix.

We must then assume at least one of $c_{23}, c_{24}, c_{34}$ must be nonzero. If $c_{23}\neq 0$, then by considering submatrices formed by rows/columns $(1,2,3,k)$ $(5\leq k\leq 8)$, we have $c_{16}=c_{17}=0$ and $c_{26}=c_{37}$. For the case that $c_{24}=0$ or $c_{34}=0$, we will also have $c_{16}=c_{17}=0$ and $c_{26}=c_{37}$ by considering appropriately chosen submatrices.

We are then left with $H$ in the form:
\begin{eqnarray*}
\resizebox{0.8\linewidth}{!}{$
H=\begin{bmatrix}
0 & 0& 0& 0& 0 &0 &0 & 0 \\
0 & 0 & c_{23} & c_{24}& c_{25}& c_{26}&0 &0  \\
0 & c_{23}^{\ast} & 0& c_{34}& c_{35}& 0& c_{26} & 0  \\
0 & c_{24}^{\ast} & c_{34}^{\ast}& 0&0 & c_{35}& c_{25}& 0 \\
0 & c_{25}^{\ast} & c_{35}^{\ast}& 0 & 0& -c_{34}& -c_{24}&0  \\
0 & c_{26}^{\ast} &0 &c_{35}^{\ast}& -c_{34}^{\ast}& 0& -c_{23}& 0  \\
0 & 0 & c_{26}^{\ast}& c_{25}^{\ast}& -c_{24}^{\ast}& -c_{23}^{\ast}&0 & 0 \\
0 & 0 &0 & 0&0& 0&0 &0
\end{bmatrix}.$}
\end{eqnarray*}

Now, recall the fact that the submatrices formed by the $k$-th and the $(9-k)$-th columns will always have rank $1$. From this it can be shown we will have $H$ is a zero matrix.

Take the submatrix formed by the second and seventh columns for example. Since they are linearly dependant, the determinant of any $2$-by-$2$ submatrix must be zero. From this we can get that $\vert c_{23}\vert^2+\vert c_{26}\vert^2=0$. Therefore $c_{23}=c_{26}=0$. By similar arguments on various submatrices, $H$ can be shown to be the zero matrix.

Thus, under our assumption that $H$ has exactly one positive eigenvalue, $\lambda_1\neq 0$. Similarly, we can also prove that $\lambda_1\neq 1, \lambda_2, \lambda_3, \lambda_4 \neq 0, 1$. We can then assume from now on that $H$ has no zero columns or rows.

Hence, there exists certain $\lambda_1, \lambda_2, \lambda_3$ and $\lambda_4 \neq 0,1$ which satisfies equation \ref{eq:lambda}.

Let us use $\Re $ and $\Im$ to denote the real part and imaginary part of a complex number. Then the above equations can be rewritten as linear equations of real numbers.

Let us use $M(\lambda_1,\lambda_2,\lambda_3,\lambda_4)$ to denote the $48$-by-$30$ coefficient matrix. If we can prove that the coefficient matrix always has rank $30$ for any $\lambda_1, \lambda_2, \lambda_3$ and $\lambda_4$, then it will imply that all $c_{ij}$'s are zeros which will  immediately contradict our assumption.

Unfortunately, we are not that lucky. $M(\lambda_1,\lambda_2,\lambda_3,\lambda_4)$ will be degenerate under certain assignment of variables $(\lambda_1,\lambda_2,\lambda_3,\lambda_4)$. For example, $rank(M(\frac{1+i}{2},\frac{1+i}{2},\frac{1+i}{2},\frac{1+i}{2}))=27<30$. However, we can still show that $M(\lambda_1,\lambda_2,\lambda_3,\lambda_4)$ will have rank $30$ except for some degenerate cases which will be dealt with separately .
The top-left $2$-by-$2$ submatrix has rank $2$ if and only if $\lambda_1\neq 0$

At least one of the following situations must happen:
\begin{enumerate}
\item [1.] $\begin{bmatrix}-C_1 & A_1\\ B_2 & C_2\end{bmatrix}$ has full rank. This implies $c_{12}=c_{17}=0$.
\item [2.] $\begin{bmatrix} A_1 & C_1\\-D_2 & A_2\end{bmatrix}$ has full rank. This implies $c_{12}=c_{28}=0$.
\item [3.] $\begin{bmatrix} B_3 & C_3\\ -C_1 & A_1\end{bmatrix}$ has full rank. This implies $c_{13}=c_{16}=0$.
\item [4.] $\begin{bmatrix}A_1 & C_1\\ -D_3 & A_3\end{bmatrix}$ has full rank. This implies $c_{13}=c_{38}=0$.
\item [5.] $\begin{bmatrix}-C_1 & A_1\\B_4 & C_4\end{bmatrix}$ has full rank. This implies $c_{14}=c_{15}=0$.
\item [6.] $\begin{bmatrix}-C_2 & A_2\\B_3 & C_3\end{bmatrix}$ has full rank. This implies $c_{23}=c_{26}=0$.
\item [7.] $\begin{bmatrix} A_2 & C_2\\ -D_3 & A_3\end{bmatrix}$ has full rank. This implies $c_{23}=c_{37}=0$.
\item [8.] $\begin{bmatrix}-C_2 & A_2\\B_4 & C_4\end{bmatrix}$ has full rank. This implies $c_{24}=c_{25}=0$.
\item [9.] $\begin{bmatrix}-C_3 & A_3\\B_4 & C_4\end{bmatrix}$ has full rank. This implies $c_{34}=c_{35}=0$.
\item [10.]
\begin{eqnarray*}
&&\det\left(\begin{bmatrix}-C_1 & A_1\\ B_2 & C_2\end{bmatrix}\right)
=\det\left(\begin{bmatrix} A_1 & C_1\\-D_2 & A_2\end{bmatrix}\right)\\
&=&\det\left(\begin{bmatrix} B_3 & C_3\\ -C_1 & A_1\end{bmatrix}\right)
=\det\left(\begin{bmatrix}A_1 & C_1\\ -D_3 & A_3\end{bmatrix}\right)\\
&=&\det\left(\begin{bmatrix}-C_1 & A_1\\B_4 & C_4\end{bmatrix}\right)
=\det\left(\begin{bmatrix}-C_2 & A_2\\B_3 & C_3\end{bmatrix}\right)\\
&=&\det\left(\begin{bmatrix} A_2 & C_2\\ -D_3 & A_3\end{bmatrix}\right)
=\det\left(\begin{bmatrix}-C_2 & A_2\\B_4 & C_4\end{bmatrix}\right)\\
&=&\det\left(\begin{bmatrix}-C_3 & A_3\\B_4 & C_4\end{bmatrix}\right)=0.
\end{eqnarray*}

With assistance of symbolic computation package like Mathematica, we find that the only solution to the above equations is $\Re \lambda_1=\Re \lambda_2=\Re \lambda_3=\Re \lambda_4=\frac{1}{2}$.
\end{enumerate}

Here we will prove that there is no Hermitian matrix in the form (~\ref{eq:form}) with only one positive eigenvalue for every situations:
\begin{enumerate}
\item [1.] $c_{12}=c_{17}=0$. Any $H$ with only one positive eigenvalue must be in the following form:

\begin{eqnarray*}
\resizebox{\linewidth}{!}{$
H=\begin{bmatrix}
0 & 0 & c_{13}& c_{14}&c_{15} &c_{16} &0 & 0 \\
0 & 0 & c_{23} & c_{24}& c_{25}& c_{26}&0 &c_{28}  \\
c_{13}^{\ast} & c_{23}^{\ast} & 0& c_{34}& c_{35}& 0& c_{37} & c_{38}  \\
c_{14}^{\ast} & c_{24}^{\ast} & c_{34}^{\ast}& 0&0 & c_{28}^{\ast}+c_{35}& c_{16}^{\ast}+c_{25}-c_{38}^{\ast}& c_{15}-c_{26}^{\ast}+c_{37}^{\ast} \\
c_{15}^{\ast} & c_{25}^{\ast} & c_{35}^{\ast}& 0 & 0& -c_{34}& -c_{24}& -c_{14}  \\
c_{16}^{\ast} & c_{26}^{\ast} &0 &c_{28}+c_{35}^{\ast}& -c_{34}^{\ast}& 0& -c_{23}& -c_{13}  \\
0 & 0 & c_{37}^{\ast}& c_{16}+c_{25}^{\ast}-c_{38}& -c_{24}^{\ast}& -c_{23}^{\ast}&0 &0 \\
0 & c_{28}^{\ast} &c_{38}^{\ast} & c_{15}^{\ast}-c_{26}+c_{37}& -c_{14}^{\ast}& -c_{13}^{\ast}&0 &0
\end{bmatrix}.
$}
\end{eqnarray*}

By considering submatrices formed by row/columns $(1,2,p,q)$ where $3\leq p<q\leq 8$, we have that the first two rows are linearly dependent. Under our assumption that there is no row of $H$ containg only zero entries, we have $c_{28}=0$.

Recall that the $4$-th and $5$-th rows are linearly dependent, thus $c_{34}(-c_{34}^{\ast})=c_{35}(c_{28}+c_{35}^{\ast})$ which now can be simplified as $\vert c_{34}\vert^2+\vert c_{35}\vert^2=0$. Hence $c_{34}=c_{35}=0$. Then

\begin{eqnarray*}
\resizebox{\linewidth}{!}{$
H=\begin{bmatrix}
0 & 0 & c_{13}& c_{14}&c_{15} &c_{16} &0 & 0 \\
0 & 0 & c_{23} & c_{24}& c_{25}& c_{26}&0 &0  \\
c_{13}^{\ast} & c_{23}^{\ast} & 0& 0& 0& 0& c_{37} & c_{38}  \\
c_{14}^{\ast} & c_{24}^{\ast} & 0& 0&0 & 0& c_{16}^{\ast}+c_{25}-c_{38}^{\ast}& c_{15}-c_{26}^{\ast}+c_{37}^{\ast} \\
c_{15}^{\ast} & c_{25}^{\ast} & 0& 0 & 0& 0& -c_{24}& -c_{14}  \\
c_{16}^{\ast} & c_{26}^{\ast} &0 &0&0& 0& -c_{23}& -c_{13}  \\
0 & 0 & c_{37}^{\ast}& c_{16}+c_{25}^{\ast}-c_{38}& -c_{24}^{\ast}& -c_{23}^{\ast}&0 &0 \\
0 &0 &c_{38}^{\ast} & c_{15}^{\ast}-c_{26}+c_{37}& -c_{14}^{\ast}& -c_{13}^{\ast}&0 &0
\end{bmatrix}.
$}
\end{eqnarray*}

Again, by applying our submatrix argument, we have the submatrix formed by $(3,4,5,6)$ columns must has rank $1$.

If there is a zero element in the submatrix formed by rows $(1,2,7,8)$ and columns $(3,4,5,6)$, then there must be a row or a column containing only zero elements in $H$. So, here we assume the submatrix formed by rows $(1,2,7,8)$ and columns $(3,4,5,6)$ does not contain any zero element.

Then $\frac{c_{15}}{c_{25}}=\frac{c_{13}}{c_{23}}=\frac{c_{38}}{c_{37}}$ which implies $c_{38}c_{25}=c_{37}c_{15}$.

Follows from the rank $1$ condition, we have
\begin{eqnarray*}
c_{15}(c_{15}^{\ast}-c_{26}+c_{37})&=&-\vert c_{14}\vert^2,\\
c_{25}(c_{16}+c_{25}^{\ast}-c_{38})&=&-\vert c_{24}\vert^2.
\end{eqnarray*}

By substituting $c_{38}c_{25}=c_{37}c_{15}$ and $c_{15}c_{26}=c_{25}c_{16}$ into the above two equations, we have
\begin{eqnarray*}
\vert c_{15}\vert^2+\vert c_{14}\vert^2 &=& c_{15}c_{26}-c_{15}c_{37}\\
&=& c_{25}c_{16}-c_{25}c_{38}\\
&=& -\vert c_{24}\vert^2-\vert c_{25}\vert^2
\end{eqnarray*}
which implies $c_{15}=c_{14}=c_{24}=c_{25}=0$. However, it contradicts our assumption that there is no zero element in the submatrix formed by $(1,2,7,8)$ rows and $(3,4,5,6)$ columns.

Similarly, we can also prove that there is no Hermtian matrix in the form ~\ref{eq:form} with only one positive eigenvalue if any of the following conditions apply.
\item [2.] $c_{12}=c_{28}=0$.
\item [3.] $c_{13}=c_{16}=0$.
\item [4.] $c_{13}=c_{38}=0$.
\item [5.] $c_{14}=c_{15}=0$.
\item [6.] $c_{23}=c_{26}=0$.
\item [7.] $c_{23}=c_{37}=0$.
\item [8.] $c_{24}=c_{25}=0$.
\item [9.] $c_{34}=c_{35}=0$.

Now, the only case we left is the following:
\item [10.]  $\Re \lambda_1=\Re \lambda_2=\Re \lambda_3=\Re \lambda_4=\frac{1}{2}$. In this case, $rank(\begin{bmatrix}-C_1 & A_1\\ B_2 & C_2\end{bmatrix})=3$. Hence $(\Re c_{12}, \Im c_{12}, \Re c_{17}, \Im c_{17})$ lies in the nullspace of $\begin{bmatrix}-C_1 & A_1\\ B_2 & C_2\end{bmatrix}=\begin{bmatrix}
-\frac{1}{2} & -b_1 & \frac{1}{2} & b_1\\
b_1 & -\frac{1}{2} & b_1 & -\frac{1}{2}\\
\frac{1}{2} & -b_2 & \frac{1}{2} & b_2 \\
b_2 & \frac{1}{2} & -b_2 & \frac{1}{2}
\end{bmatrix}$.  Thus 
\begin{eqnarray*}
&&[{c_{12}}:{c_{17}}]\\
&=&[{2(b_2-b_1)+(1+4b_1b_2)\emph{i}}:{2(b_1+b_2)+(4b_1b_2-1)\emph{i}}].
\end{eqnarray*}

Similarly, we will have
\begin{eqnarray}\label{eq:ratio}
&&\lbrack{c_{12}}:{c_{17}}:{c_{28}}\rbrack\nonumber\\
&=& \lbrack{2(b_2-b_1)+(1+4b_1b_2)\emph{i}}:\nonumber\\
&&{2(b_1+b_2)+(4b_1b_2-1)\emph{i}}: {-2(b_1+b_2)-(4b_1b_2-1)\emph{i}}\rbrack;\nonumber\\
&&\lbrack{c_{13}}:{c_{16}}:{c_{38}}\rbrack\nonumber\\
&=& \lbrack{2(b_1-b_3)-(1+4b_1b_3)\emph{i}}:\nonumber\\
&&{-2(b_1+b_3)-(4b_1b_3-1)\emph{i}}:{2(b_1+b_3)+(4b_1b_3-1)\emph{i}}\rbrack;\nonumber\\
&&\lbrack c_{23}:c_{26}:c_{37}\rbrack\nonumber\\
&=& \lbrack2(b_3-b_2)+(4b_2b_3+1)\emph{i}: \nonumber\\
&&2(b_2+b_3)+(4b_2b_3-1)\emph{i} :-2(b_2+b_3)-(4b_2b_3-1)\emph{i}\rbrack;\nonumber\\
&&\lbrack c_{14}:c_{15} \rbrack\nonumber\\
&=&\lbrack2(b_4-b_1)-(4b_1b_4+1)\emph{i}:\nonumber\\
&&2(b_1+b_4)+(4b_1b_4-1)\emph{i}\rbrack;\nonumber\\
&&\lbrack c_{24}:c_{25} \rbrack \nonumber\\
&=& \lbrack 2(b_4-b_2)+(4b_2b_4+1)\emph{i}:\nonumber\\
&& 2(b_2+b_4)+(4b_2b_4-1)\emph{i}\rbrack;\nonumber\\
&&\lbrack c_{34}:c_{35}\rbrack \nonumber\\
&=& \lbrack 2(b_4-b_3)+(4b_3b_4+1)\emph{i}:\nonumber\\
&&2(b_3+b_4)+(4b_3b_4-1)\emph{i}\rbrack. \nonumber
\end{eqnarray}

Here $\lbrack \ q_1:q_2:\cdots : q_m \rbrack =\lbrack r_1+s_1\emph{i} :r_2+s_2 \emph{i}:\cdots : r_m+s_m\emph{i} \rbrack$ means there exists some $\mu \in \mathbb{R}$ such that $q_i=\mu (r_i+s_i\emph{i})$ for any $1\leq i\leq m$.

Observe that $c_{28}=-c_{17}, c_{38}=-c_{16}, c_{37}=-c_{26}$, we thus simplify the matrix form of $H$ as the following:
\begin{eqnarray*}
\resizebox{\linewidth}{!}{$
H=\begin{bmatrix}
0 & c_{12}& c_{13}& c_{14}&c_{15} &c_{16} &c_{17} & 0 \\
c_{12}^{\ast} & 0 & c_{23} & c_{24}& c_{25}& c_{26}&0 & -c_{17}  \\
c_{13}^{\ast} & c_{23}^{\ast} & 0& c_{34}& c_{35}& 0& -c_{26} & -c_{16}  \\
c_{14}^{\ast} & c_{24}^{\ast} & c_{34}^{\ast}& 0&0 & c_{35}-2c_{17}^{\ast}& c_{25}+2c_{16}^{\ast}& c_{15}-2c_{26}^{\ast} \\
c_{15}^{\ast} & c_{25}^{\ast} & c_{35}^{\ast}& 0 & 0& -c_{34}& -c_{24}& -c_{14}  \\
c_{16}^{\ast} & c_{26}^{\ast} &0 & c_{35}^{\ast}-2c_{17}& -c_{34}^{\ast}& 0& -c_{23}& -c_{13}  \\
c_{17}^{\ast} & 0 & -c_{26}^{\ast}& c_{25}^{\ast}+2c_{16}& -c_{24}^{\ast}& -c_{23}^{\ast}&0 & -c_{12} \\
0 & -c_{17}^{\ast} & -c_{16}^{\ast} & c_{15}^{\ast}-2c_{26}& -c_{14}^{\ast}& -c_{13}^{\ast}&-c_{12}^{\ast} &0
\end{bmatrix}.
$}
\end{eqnarray*}

It follows from the fact that the submatrix formed by $4$-th and $5$-th columns has rank exactly $1$, we have $c_{14}(-c_{14}^{\ast})=c_{15}(c_{15}^{\ast}-2c_{26})$.
Thus at least one of the following cases must happen:
\begin{enumerate}
\item [(10.1)] $c_{14}=c_{15}=0$. We can still assume there is no column containing only zero elements as this is the case that we have already discussed. Thus $c_{26}=0$ which would also lead to $c_{23}=0$.
\item [(10.2)] $c_{26}=c_{15}^{\ast}$.
\end{enumerate}

Similarly, at least one of the following conditions:
\begin{enumerate}
\item [(10.I)] $c_{24}=c_{25}=c_{16}=c_{13}=0$; or
\item [(10.II)] $c_{16}=-c_{25}^{\ast}$
\end{enumerate}
and one of the following conditions:
\begin{enumerate}
\item [(10.A)] $c_{34}=c_{35}=c_{17}=c_{12}=0$; or
\item [(10.B)] $c_{17}=c_{35}^{\ast}$
\end{enumerate}
must apply.

We have already discussed the cases that $c_{12}=c_{17}=0$, $c_{13}=c_{16}=0$ or $c_{23}=c_{26}=0$ previously. Hence the only remaining case is $c_{26}=c_{15}^{\ast}, c_{16}=-c_{25}^{\ast}, c_{17}=c_{35}^{\ast}$. Thus
\begin{eqnarray*}
H=\begin{bmatrix}
0 & c_{12}& c_{13}& c_{14}&c_{15} &-c_{25}^{\ast} &c_{35}^{\ast} & 0 \\
c_{12}^{\ast} & 0 & c_{23} & c_{24}& c_{25}& c_{15}^{\ast}&0 & -c_{35}^{\ast}  \\
c_{13}^{\ast} & c_{23}^{\ast} & 0& c_{34}& c_{35}& 0& -c_{15}^{\ast} & c_{25}^{\ast}  \\
c_{14}^{\ast} & c_{24}^{\ast} & c_{34}^{\ast}& 0&0 & -c_{35}& -c_{25}& -c_{15} \\
c_{15}^{\ast} & c_{25}^{\ast} & c_{35}^{\ast}& 0 & 0& -c_{34}& -c_{24}& -c_{14}  \\
-c_{25} & c_{15} &0 & -c_{35}^{\ast}& -c_{34}^{\ast}& 0& -c_{23}& -c_{13}  \\
c_{35} & 0 & -c_{15}& -c_{25}^{\ast}& -c_{24}^{\ast}& -c_{23}^{\ast}&0 & -c_{12} \\
0 & -c_{35} & c_{25} & -c_{15}^{\ast}& -c_{14}^{\ast}& -c_{13}^{\ast}&-c_{12}^{\ast} &0
\end{bmatrix}.
\end{eqnarray*}

According to $c_{26}=c_{15}^{\ast}$, we have $2(b_2+b_3)(1-4 b_1b_4)=(4b_2b_3-1) (2b_1+2b_4)$ which implies $4(b_1b_2b_3+b_1b_2b_4+b_1b_3b_4+b_2b_3b_4)=b_1+b_2+b_3+b_4$.

\begin{enumerate}
\item [1.] $4b_1b_2+4b_1b_3+4b_2b_3=1$. Thus $b_1+b_2+b_3=4b_1b_2b_3$. However, one can easy to verify that there do not exist three real numbers $b_1,b_2,b_3$ satisfying these two equations.
\item [2.] { $4b_1b_2+4b_1b_3+4b_2b_3\neq 1$. Hence $b_4=\frac{b_1+b_2+b_3-4b_1b_2b_3}{4b_1b_2+4b_1b_3+4b_2b_3-1}$. By substituting the assignment of $b_4$ into Equation~\ref{eq:ratio}, we have
\begin{eqnarray*}
c_{14}&=&p \cdot \bigg( \frac{2(1-4b_1^2)(b_2+b_3)+4b_1(1-4b_2b_3)}{4b_1b_2+4b_1b_3+4b_2b_3-1}\\
&&+\frac{-8b_1(b_2+b_3)+(1-4b_1^2)(1-4b_2b_3)}{4b_1b_2+4b_1b_3+4b_2b_3-1}\emph{i}\bigg);\\
c_{15}&=&p \cdot\bigg(\frac{2(1+4b_1^2)(b_2+b_3)}{4b_1b_2+4b_1b_3+4b_2b_3-1}\\
&&+\frac{(1+4b_1^2)(1-4b_2b_3)}{4b_1b_2+4b_1b_3+4b_2b_3-1}\emph{i}\bigg);\\
c_{23}&=&p \cdot \bigg(\frac{2(1+4b_1^2)(b_3-b_2)}{4b_1b_2+4b_1b_3+4b_2b_3-1}\\
&&+\frac{(1+4b_1^2)(1+4b_2b_3)}{4b_1b_2+4b_1b_3+4b_2b_3-1}\emph{i}\bigg);\\
c_{24}&=& q \cdot \bigg(\frac{2(1-4b_2^2)(b_1+b_3)+4b_2(1-4b_1b_3)}{4b_1b_2+4b_1b_3+4b_2b_3-1}\\
&&+\frac{8b_2(b_1+b_3)-(1-4b_2^2)(1-4b_1b_3)}{4b_1b_2+4b_1b_3+4b_2b_3-1}\emph{i}\bigg);\\
c_{25}&=& q \cdot \bigg(\frac{2(1+4b_2^2)(b_1+b_3)}{4b_1b_2+4b_1b_3+4b_2b_3-1}\\
&&+\frac{(1+4b_2^2)(1-4b_1b_3)}{4b_1b_2+4b_1b_3+4b_2b_3-1}\emph{i}\bigg);\\
c_{13}&=& q \cdot\bigg(\frac{2(1+4b_2^2)(b_1-b_3)}{4b_1b_2+4b_1b_3+4b_2b_3-1}\\
&&-\frac{(1+4b_2^2)(1+4b_1b_3)}{4b_1b_2+4b_1b_3+4b_2b_3-1}\emph{i}\bigg);\\
c_{34}&=&r \cdot\bigg(\frac{2(1-4b_3^2)(b_1+b_2)+4b_3(1-4b_1b_2)}{4b_1b_2+4b_1b_3+4b_2b_3-1}\\
&&+\frac{8b_3(b_1+b_2)-(1-4b_3^2)(1-4b_1b_2)}{4b_1b_2+4b_1b_3+4b_2b_3-1}\emph{i}\bigg);\\
c_{35}&=&r \cdot\bigg(\frac{2(1+4b_3^2)(b_1+b_2)}{4b_1b_2+4b_1b_3+4b_2b_3-1}\\
&&+\frac{(1+4b_3^2)(1-4b_1b_2)}{4b_1b_2+4b_1b_3+4b_2b_3-1}\emph{i}\bigg);\\
c_{12}&=&r \cdot\bigg(\frac{2(1+4b_3^2)(b_2-b_1)}{4b_1b_2+4b_1b_3+4b_2b_3-1}\\
&&+\frac{(1+4b_3^2)(1+4b_1b_2)}{4b_1b_2+4b_1b_3+4b_2b_3-1}\emph{i}\bigg).\\
\end{eqnarray*}

Again, with assistance of symbolic computation package like Mathematica, we can verify that the characteristic polynomial of $H$ contains only even powers. This implies $H$ has nonzero eigenvalue $\lambda$ if and only if it also has eigenvalue $-\lambda$. Therefore, under our assumption that $H$ has only $1$ positive eigenvalue, $H$ also has only $1$ negative eigenvalue.

However, let us consider the $3$-by-$3$ submatrix of $H$ formed by $(5,7,8)$-th rows and $(1,2,3)$-th columns . Its determinant is $(-\emph{i} + 2 b_1) (\emph{i}+ 2 b_1)^2 (1 + 2 \emph{i} b_2) (\emph{i} + 2 b_2)^2 (\emph{i} - 2 b_3)^2 (\emph{i} + 2 b_3)^2 r ((1 + 4 b_1^2) p^2 + (1 + 4 b_2^2) q^2 + (1 + 4 b_3^2) r^2)$. It is always nonzero unless $p=q=r$ or $r=0$. If $r=0$ this implies that $c_{12}=c_{34}=c_{35}=0$. This case has already been covered, Thus $H$ has rank at least $3$ which contradicts our previous conclusion that $H$ has only one positive eigenvalue and only one negative eigenvalue.}
\end{enumerate}
\end{enumerate}

To summarize, under our assumption that $H$ has only one positive eigenvalue, a contradiction always exist in every situation we studied. Hence, $H$ must has at least $2$ positive eigenvalues and at least $2$ negative eigenvalues. This completes our proof.
\end{proof}

\end{document}